\documentclass[1p,times]{elsarticle}

\usepackage{amsmath,amssymb}
\usepackage{amsthm}
\usepackage{enumitem}
\usepackage{float}

\usepackage{graphicx}
\graphicspath{{./}}
\usepackage{changepage}
\usepackage[table]{xcolor}
\usepackage{array}

\usepackage{algorithmic}
\usepackage{mdframed}

\newcounter{spideralg}
\newcommand{\algcaption}[1]{\refstepcounter{spideralg}%
  \par\noindent\textbf{Algorithm \thespideralg.} #1\par\vspace{3pt}}

\usepackage{url}
\usepackage{hyperref}
\hypersetup{hidelinks}

\theoremstyle{plain}
\newtheorem{theorem}{Theorem}
\newtheorem{proposition}{Proposition}

\theoremstyle{definition}
\newtheorem{definition}{Definition}
\newtheorem{condition}{Condition}
\newtheorem{assumption}{Assumption}
\newtheorem{problem}{Problem}

\newcommand{\C}{\mathbb{C}}

\newcommand{\Herm}{{}^{H}}          
\newcommand{\T}{{}^{\top}}

\newcommand{\diag}{\operatorname{diag}}

\journal{NeuroImage}

\begin{document}

\begin{frontmatter}

\title{SPIDER: Recovering brain-wide directed information flow from incomplete and asynchronous neural recordings}

\author[tsinghua]{Yisi S. Zhang\corref{cor}}
\ead{zhangyisi@tsinghua.edu.cn}

\author[ufrn]{Daniel Y. Takahashi\corref{cor}}
\ead{takahashiyd@gmail.com}

\cortext[cor]{Corresponding authors.}

\address[tsinghua]{Department of Psychological and Cognitive Sciences, Tsinghua University, Beijing, China}
\address[ufrn]{Brain Institute, Federal University of Rio Grande do Norte, Natal, Brazil}

\begin{abstract}
Mapping the directed flow of information between brain regions---their effective connectivity---is central to understanding brain function, yet large-scale recordings sample only a fraction of the brain at a time: sessions, animals, and laboratories cover different, partially overlapping regions, usually without a shared temporal reference. Established directed-connectivity methods (Granger causality, dynamic causal modeling, partial directed coherence, PDC) require all regions to be recorded simultaneously and with a common clock. We introduce SPIDER (Stitched Power-spectra for Inferring Directed information flow from incomplete and asynchronous Experimental Recordings), a non-parametric, frequency-domain framework that recovers directed information flow from such incomplete, asynchronous recordings: it stitches local power-spectral estimates from overlapping channel subsets into a global spectral matrix and obtains frequency-resolved directed interactions by canonical spectral factorization and PDC, without temporal alignment, while nuclear-norm completion fills in never-co-observed region pairs. With consistency guarantees, we validate SPIDER on simulations, two-photon calcium imaging, and the International Brain Laboratory Neuropixels dataset, recovering directed flow among 50 areas from 43 sessions in 12 laboratories---never recorded together---with the strongest interactions matching Allen atlas projections (AUC $\approx 0.9$). Beyond validation, SPIDER reveals what no single recording can: brain-wide spontaneous flow is largely recurrent, but in the theta band it forms a significant feedforward hierarchy with the hippocampal formation at its source. Applied to resting human intracranial EEG (43 patients, partially overlapping coverage), it recovers the same theta-band hierarchy across species and modality. SPIDER makes whole-brain effective-connectivity analysis tractable for multi-session, multi-subject datasets previously incompatible with directed-flow inference.
\end{abstract}


\end{frontmatter}


\section{Introduction}
A central goal of systems and cognitive neuroscience is to map how information flows between brain regions, and crucially in which direction~\cite{friston2003dynamic,seth2015granger,bastos2016tutorial}. This directed, or \emph{effective}, connectivity is increasingly studied with large-scale recordings that, by physical necessity, sample only a fraction of the brain at any one time. High-density electrode arrays such as Neuropixels probes reach only a handful of regions per insertion~\cite{international2025reproducibility, international2025brain}, and different sessions or animals cover different, partially overlapping sets of regions; sequential two-photon calcium imaging likewise records overlapping neuronal populations across separate fields of view~\cite{microns2025functional}. The full network is never observed simultaneously, and during spontaneous activity no external trigger is available to align recordings in time. Recovering the directed interactions that govern such systems from these fragmented, asynchronous observations is the problem we address.

Although our motivation and validation are firmly in neuroscience, the underlying difficulty is general: across many complex systems, directed information flow among a high-dimensional set of components must be inferred from data, a task central to causal discovery~\cite{granger1969investigating}, system identification~\cite{ljung1998system}, and representation learning~\cite{scholkopf2021toward}, yet physical, temporal, or bandwidth constraints---limited channels, experimental time, or space for sensors---routinely make it impossible to observe every component~\cite{soudry2015efficient}. The brain is the setting where these constraints are most acute and where the payoff of solving them is largest, and it is the setting we develop the method for throughout.

This incompleteness poses a fundamental challenge for causal inference. When only a subset of variables is observed, unobserved components can act as latent confounders, inducing spurious dependencies and obscuring the true interaction structure---a well-known problem in graphical models and causal discovery~\cite{hsiao1982autoregressive,eichler2012causal}. For example, in a three-node network, observing only two nodes may lead to the erroneous conclusion that they directly interact when, in fact, their apparent dependence is mediated by an unobserved third node~\cite{das2020systematic} (Fig~\ref{fig1}B). While alignment to a common external event or experimental trigger can sometimes mitigate this issue~\cite{qiao2020causal}, many settings of interest, such as spontaneous neural activity~\cite{stringer2019spontaneous} and naturalistic animal behavior and state~\cite{calhoun2019unsupervised,zhang2022active}, lack any reliable synchronization signal. In such cases, the global joint distribution of all variables is never directly sampled, and cross-dependencies between non-co-observed subsets remain undefined.

A natural strategy is to pool information across multiple partial observations so that, collectively, most variables are observed, albeit never all at once (Fig~\ref{fig1}A: the problem). Current strategies to stitch these fragmented observations typically fall into two categories, both with significant limitations. Parametric approaches~\cite{turaga2013inferring,soudry2015efficient,kang2024integration} rely on rigid modeling assumptions (e.g., generalized linear models or vector autoregressive (VAR) models) and are not adequate for extracting frequency-domain information, which is critical for several applications such as neuroscience, where oscillatory interactions in distinct frequency bands carry functionally distinct signals~\cite{buzsaki2006rhythms}. Covariance-based aggregation methods~\cite{bishop2014deterministic,vinci2019graph,chang2023nonparanormal} remain restricted to undirected dependencies and cannot recover directional interactions. As a result, without restrictive conditions, directed information flow remains difficult to infer for general signals in high-dimensional systems with incomplete and asynchronous observations.

In this work, we introduce SPIDER, a non-parametric framework that overcomes these limitations by operating in the frequency domain. Our key idea is to first stitch local power spectral density (PSD) estimates, which are computed independently on partially overlapping subsets of variables, into a global spectral representation. Second, we use the fact that minimum-phase wide-sense stationary processes are completely characterized by their spectral representation, enabling directed inference without requiring temporal alignment~\cite{wilson1978convergence}. This strategy enables principled handling of incompleteness and high dimensionality. Leveraging non-parametric partial directed coherence (PDC)~\cite{baccala2001partial}, we infer directed information flow from the stitched spectral density, avoiding the need for explicit time-domain models, latent state assumptions, or synchronized recordings. Together, these components extend directed information flow inference to a broad class of settings where full observability, low dimensionality, and temporal alignment cannot be assumed. From the application perspective, frequency-domain measures are particularly well suited to neuroscience and biomedical signal analysis, and to any situation in which understanding information flow at different temporal scales matters: band-specific effective connectivity (for example, theta- versus gamma-range interactions) is a natural and widely used readout in the analysis of neural recordings.

\subsection{Contributions}
This work makes two main contributions, both aimed at extending directed effective-connectivity analysis to the multi-session, multi-animal recordings that dominate modern neuroscience.

\textbf{(1) A non-parametric estimator of whole-network, frequency-domain directed connectivity under asynchronous block recordings.} We introduce a two-stage procedure that combines spectral stitching and canonical factorization to recover directed information flow when the full set of variables (e.g., brain areas) is never observed simultaneously, provided every pair is co-observed in at least one recording block. The estimator is non-parametric, applies to any wide-sense stationary process with bounded spectral density (including infinite-order VAR processes, autoregressive moving average (ARMA) processes, and multivariate Hawkes point processes such as those describing spike trains), and does not require temporal alignment across blocks. Theorem~\ref{thm:complete_short} establishes uniform consistency, and we identify the high-dimensional spectral consistency result of~\cite{zhang2025spectral} as a concrete sufficient condition under which the estimator applies in regimes where the number of variables $K$ greatly exceeds the per-block sample size $T_u$, with $K$ allowed to grow nearly exponentially in $T_u$.

\textbf{(2) Extension to settings with unobserved variable pairs.} When some pairs are never co-observed---the typical situation when pooling across electrode placements that vary from animal to animal---we incorporate a structured matrix completion step (nuclear-norm minimization on the augmented Hermitian form) before factorization. Theorem~\ref{thm:incomplete} states the resulting consistency result conditionally on completion success, and Proposition~\ref{prop:nnm-sufficient} provides explicit sufficient conditions based on spikiness, observation coverage, and near-low-rank structure under which the conditional hypothesis is satisfied. The near-low-rank assumption is justified by a generative mixing-model argument and is consistent with the empirically observed low-dimensional structure of large-scale neural population activity~\cite{cunningham2014dimensionality}.

\paragraph{Related work.}
Inferring directed connectivity from partially observed multivariate time series has received considerable attention, particularly in neuroscience, where limited recording capacity and subsampling are inherent constraints. One class of methods employs explicit generative models, including time-domain latent dynamical systems~\cite{turaga2013inferring}, generalized linear models for point processes~\cite{pillow2008spatio,soudry2015efficient}, neural dynamical models and dynamic causal modelling~\cite{friston2003dynamic,kang2024integration}, and latent high-dimensional vector-autoregressive models for Granger causality~\cite{fan2020estimation}. While powerful when their assumptions are satisfied, these methods rely on specific generative models and often require carefully designed experimental protocols~\cite{soudry2015efficient}, limiting their applicability when dynamics are complex, heterogeneous, or collected across asynchronous sessions without a shared temporal reference. Our approach, in its most general form, assumes only the existence of an invertible spectral representation, which is satisfied in most applications.

A second, related line of work focuses on covariance and graph reconstruction from partial observations~\cite{bishop2014deterministic,vinci2019graph,chang2023nonparanormal}. These methods primarily recover a symmetric positive semi-definite matrix---typically a covariance or precision matrix---from pairwise measurements observed simultaneously in at least some subset of the data, with deterministic missingness patterns and observations potentially corrupted by noise. Two distinctions from our approach are worth highlighting. First, previous methods are designed for undirected graph inference and do not incorporate directionality; our framework includes a spectral factorization step that enables estimation of directed information flow between node pairs. Second, our method imposes a positive definiteness constraint on the PSD matrix, essential for ensuring a valid non-parametric PDC estimate. This requirement constrains the class of admissible completions and distinguishes our problem from covariance completion, where strict low-rank assumptions are often appropriate~\cite{bishop2014deterministic,liu2016low}.

Finally, we contrast with traditional PDC estimation. Classical PDC proceeds by fitting a finite-order VAR model, often using algorithms such as the Nuttall--Strand recursion, and then transforming the estimated coefficients into the frequency domain~\cite{baccala2007generalized}. This approach is inherently parametric and depends critically on correct model order selection and the adequacy of a finite-order VAR representation; multivariate point processes or mixed continuous and point processes cannot be modeled directly in this framework. In contrast, our method is fully non-parametric: as in earlier non-parametric evaluations of PDC~\cite{amblard2015nonparametric}, PDC is obtained directly from the spectral density matrix via canonical spectral factorization, without assuming a finite-order VAR model. The spectral density itself can be estimated using non-parametric techniques such as multi-taper spectral analysis~\cite{thomson2005spectrum}. This greatly broadens applicability to non-Gaussian data and point-process signals such as neural spike trains, where parametric assumptions may be inappropriate or ill-defined. In short, the method combines the directionality absent from covariance-based graph quilting with the fully non-parametric spectral formulation absent from VAR-based PDC, making it possible to estimate directed connectivity from asynchronous, partially overlapping recordings.


\section{Materials and methods}

In brief, the method proceeds in three main steps (Fig~\ref{fig1}C-D). \emph{(1) Stitching:} the cross-spectrum between any pair of variables is estimated from whichever recording blocks contain both, and these local estimates are averaged into a single global spectral matrix---no block need contain all variables, and the blocks need not be aligned in time. \emph{(2) Completion (optional):} if some pairs are never co-observed, the missing entries of the spectral matrix are filled in using principled matrix completion. \emph{(3) Factorization:} the completed matrix is factorized and turned into partial directed coherence (PDC), yielding the directed information flow of the whole network. The remainder of this section defines each step and states the conditions under which it recovers the true whole-network PDC.

\subsection{Directed information flow in the frequency domain}

We quantify directed interactions using partial directed coherence (PDC)~\cite{baccala2001partial,takahashi2010information}, a frequency-domain measure of the direct, directional influence of one time series on another at each frequency, after accounting for all other recorded series. PDC was originally defined through a vector autoregressive (VAR) model, but---as we use it here---it depends only on the spectral density of the process and its canonical (minimum-phase) factorization. It is therefore fully non-parametric and applies to any wide-sense stationary process with a bounded spectral density: continuous signals (e.g.\ local field potentials, calcium fluorescence), point processes (spike trains, through the Bartlett spectrum), and mixtures of the two. The VAR derivation, the equivalent non-parametric form, and the precise conditions on the spectrum are collected in Appendix~\ref{S_mathdetails}.

For a $K$-dimensional process with autoregressive transfer matrix $\bar{A}(\omega)$ (Appendix~\ref{S_mathdetails}) and innovation covariance $\Sigma$, the informational PDC from series $\ell$ to series $k$ is the normalized transfer function
\begin{equation} \label{eq:pdc}
\pi_{k\leftarrow \ell}(\omega) = \frac{\bar{A}_{k\ell}(\omega)\,\Sigma^{-1/2}_{kk}}{\sqrt{\bar{A}_{\cdot \ell}(\omega)^H \Sigma^{-1} \bar{A}_{\cdot \ell}(\omega)}},
\end{equation}
where $\bar{A}_{\cdot \ell}(\omega) = [\bar{A}_{1\ell}(\omega),\ldots,\bar{A}_{K\ell}(\omega)]^T$ and the superscript $H$ denotes the conjugate transpose. The PDC matrix $\Pi(\omega)$ has entries $\pi_{k\leftarrow \ell}(\omega)$.

This measure has the property that if $\pi_{k\leftarrow \ell}(\omega) = 0$ for all $\omega \in [-\pi, \pi)$, then $A_{k\ell}(s) = 0$ for all $s \geq 1$. If $\{X(t)\}_{t \in\mathbb{Z}}$ is Gaussian, the nullity of $\pi_{k\leftarrow \ell}(\omega)$ at all frequencies is equivalent to the conditional independence of $X_k(t)$ and the past of $X_\ell$ given the past of all remaining time series. In the Gaussian case, the integrated quantity
\begin{equation} \label{eq:infoflow}
M_{k\leftarrow \ell} = -\frac{1}{4\pi}\int_{-\pi}^{\pi}\log\!\bigl(1-|\pi_{k\leftarrow \ell}(\omega)|^2\bigr)\,d\omega
\end{equation}
admits a direct information-theoretic interpretation as the directed information flow from $X_\ell$ to $X_k$ conditional on all other time series~\cite{gelfand1959calculation,takahashi2010information}.

Crucially, the validity of PDC depends on the \emph{full} multivariate autoregressive representation of the system. If the VAR model in~\eqref{eq:AR} is estimated from an incomplete set of time series, the resulting values can differ substantially from the PDC obtained from the complete set. Unobserved variables may induce spurious directed interactions or mask existing ones through cancellation~\cite{baccala2001partial}, motivating the need for principled methods that recover PDC under partial observability.

\subsection{Observation model and the stitching problem}

Let $[K] := \{1,\ldots,K\}$ index a $K$-dimensional system, and let $\mathbf{X} = \{X(t)\}_{t\in\mathbb{Z}}$ be a $K$-dimensional wide-sense stationary stochastic process. We observe $\mathbf{X}$ through $m$ recording blocks. The $u$-th block consists of a subset $\Omega_u \subseteq [K]$ of variables observed jointly during a contiguous time window $[t_u^{\mathrm{start}},\,t_u^{\mathrm{start}} + T_u]$ of length $T_u$. Different blocks need not overlap in time; their observation windows may be disjoint, and we do not assume any shared temporal reference or external trigger that would allow blocks to be aligned. We assume throughout that $\bigcup_{u=1}^m \Omega_u = [K]$, so every variable is observed in at least one block. We refer to the PDC computed from the fully observed process $\mathbf{X}$ as the \emph{whole-network PDC}.

We will refer to the following condition.
\begin{condition}[Completeness] \label{cond:complete}
The collection of observation blocks $\{\Omega_u\}_{u=1}^m$ satisfies the \emph{completeness condition} if, for every pair $(i,j) \in [K] \times [K]$, there exists $u \in [m]$ such that $\{i,j\} \subseteq \Omega_u$.
\end{condition}
In words, completeness requires every pairwise interaction to be observed in at least one block, but no single block need contain all variables.

\medskip\noindent We address two problems concerning the recovery of whole-network PDC from partial, asynchronous observations.

\begin{problem}[Stitching under completeness]\label{prob:complete}
Suppose every pair of variables is jointly observed in at least one block (Condition~\ref{cond:complete}), although the full set is never observed simultaneously. Given the asynchronous block observations $\{X_{\Omega_u}(t)\}_{u \in [m],\, t \in [T_u]}$, can the whole-network PDC be recovered consistently as $\min_u T_u \to \infty$?
\end{problem}

\begin{problem}[Stitching with missing pairs]\label{prob:incomplete}
Suppose $\{\Omega_u\}_{u=1}^m$ does \emph{not} satisfy Condition~\ref{cond:complete}: there exist pairs $(i,j) \in [K] \times [K]$ such that no block contains both $i$ and $j$. Some pairwise interactions are therefore never observed. Can the whole-network PDC still be recovered, and under what additional conditions on the observation pattern and the underlying process?
\end{problem}

Fig~\ref{fig2}A\&F illustrates the two problem settings. In Problem~\ref{prob:complete}, recordings cover overlapping subsets of variables that collectively span all pairs; in Problem~\ref{prob:incomplete}, some pairs lie outside the union of co-observed sets.

\subsection{The SPIDER pipeline}

Our approach to both problems consists of a frequency-domain pipeline whose core logic is straightforward (Fig~\ref{fig1}C-D). We outline the steps here at the level of intuition; formal definitions and consistency results follow in the subsequent subsections.

The first step, \textbf{stitching}, addresses the core challenge of asynchronous recording. Although no single block captures all $K$ channels simultaneously, the cross-spectrum $S_{ij}(\omega)$ between any pair $(i,j)$ can be estimated from any block that contains both channels. Averaging these block-level estimates across all blocks where each pair appears reconstructs the entries of the global power spectral density (PSD) matrix $S(\omega)$ without ever requiring simultaneous observation (Fig~\ref{fig1}C). This averaging is well posed because each cross-spectrum $S_{ij}(\omega)$ is a marginal property of the pair $(i,j)$, unchanged by which other channels happen to share its block; the directed structure, by contrast, is a global property of the assembled matrix, so stitching aggregates marginal quantities and evaluates the global functional only once. The further fact enabling recovery is that, for a wide-sense stationary purely non-deterministic process whose spectrum is bounded away from zero, the PSD matrix uniquely determines the canonical (minimum-phase) spectral factorization, and therefore the whole-network PDC. Stitching thus bypasses the need for synchronization, recovering the full directed structure from independently recorded, asynchronous blocks (Fig~\ref{fig1}D).

The second step, \textbf{matrix completion}, is required when Condition~\ref{cond:complete} fails and some entries of $S(\omega)$ are not recovered by stitching (dashed entries in Fig~\ref{fig1}C). Because PDC is a global property of the system, missing entries cannot simply be ignored: completing the matrix from partial entries is ill-posed and requires a structural prior linking the unobserved entries to the observed ones (Fig~\ref{fig1}B\&D). We adopt a near-low-rank prior on $S(\omega)$, motivated by a generative mixing model in which population activity is driven by a small number of shared latent sources (Appendix~\ref{S_mathdetails}). This keeps $S(\omega)$ full rank but dominated by a few singular values---the standard target of nuclear-norm minimization (NNM) completion---consistent with the well-established low-dimensional structure of large-scale neural population activity~\cite{cunningham2014dimensionality}.

The third step, \textbf{canonical spectral factorization and calculation of PDC}. Applying Wilson's algorithm~\cite{wilson1978convergence} to the stitched $S(\omega)$ yields the transfer function matrix $H(\omega)$ and innovation covariance $\Sigma$ such that $S(\omega) = H(\omega)\,\Sigma\, H(\omega)^H$. Setting $\bar{A}(\omega) = H(\omega)^{-1}$ recovers the VAR transfer function, from which the whole-network PDC is obtained directly via~\eqref{eq:pdc}.

Two refinements address the additional demands of high-dimensional, multi-channel recordings; both are detailed where they are first used. When many channels, such as the neurons recorded within one brain area, share a category, spectral principal component analysis (PCA) replaces each group by a single canonical spectral mode before stitching, so that groups remain comparable across sessions that sampled different channels (Fig~\ref{fig1}E; \nameref{sec:areapdc}). And because finite-sample stitched matrices need not be positive definite when $K$ is large, we regularize the inverse spectrum with a complex Hermitian graphical lasso (GLASSO), which guarantees a valid factorization and reduces variance when the underlying dependence structure is sparse (Fig~\ref{fig1}F; \nameref{sec:completeobs}).

The full procedure is summarized in Algorithm~\ref{alg:stitched_pdc} (Appendix~\ref{algorithm1}).

\begin{figure}[h!]
  \centering
  \includegraphics[width=0.95\textwidth]{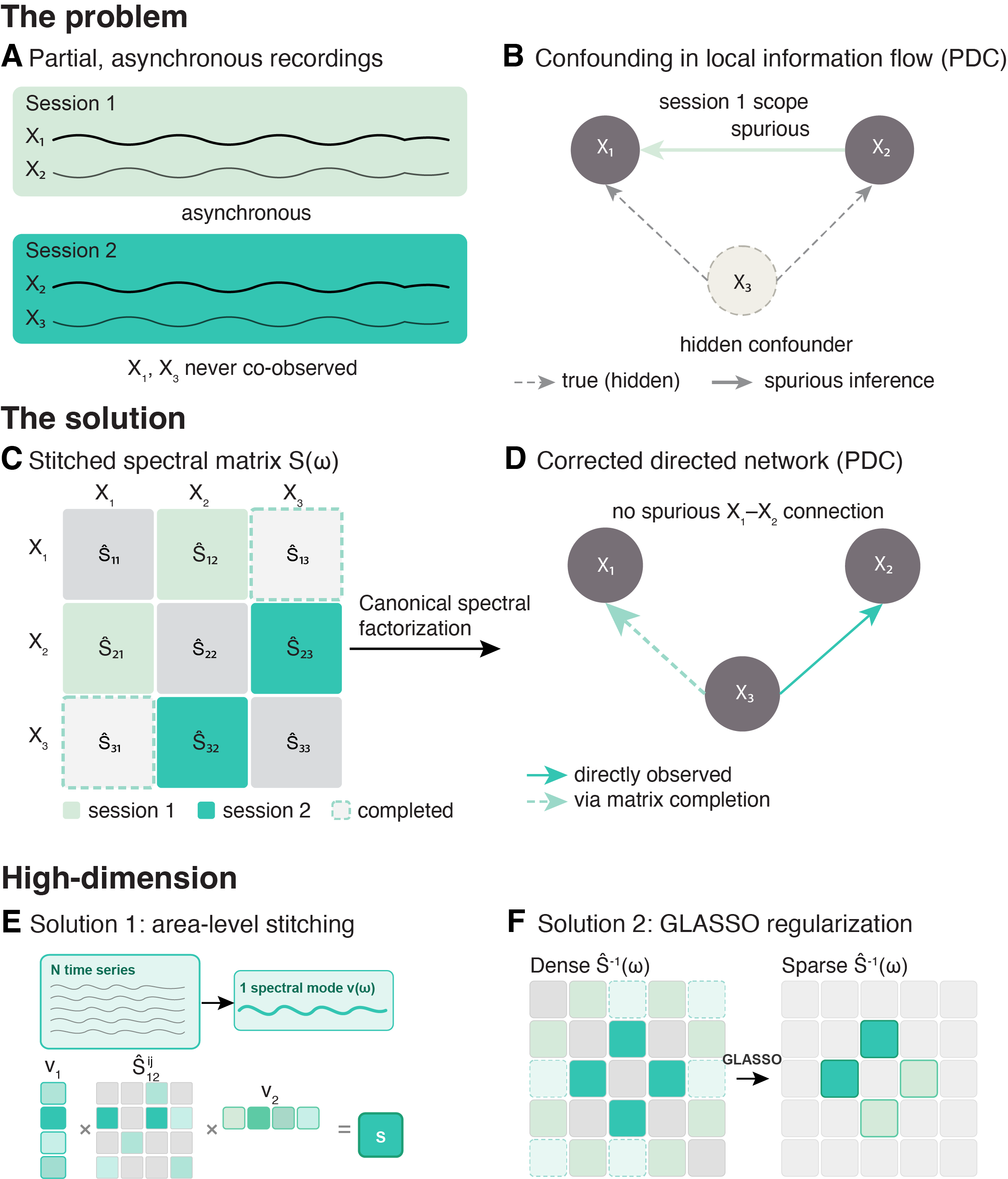}
  \vspace{0.5cm}
  \caption{{\bf Problem setting and pipeline overview.}
(A) Recording sessions cover overlapping but incomplete subsets of a network.
(B) When PDC is estimated within a single session's incomplete scope, a hidden common driver $X_3$ induces a spurious directed connection between $X_1$ and $X_2$ that is indistinguishable from true coupling under local observation.
(C) Stitched spectral matrix $\hat{S}(\omega)$. Local cross-spectral estimates from session 1 and session 2 are assembled into a global matrix. Entries that are never directly co-observed can be inferred via matrix completion.
(D) Applying PDC to the stitched spectral matrix should eliminate the spurious $X_1$--$X_2$ connection and recover the true directed information flow. Solid arrows indicate connections derived from directly observed cross-spectra; dashed arrows indicate connections inferred via matrix completion.
(E) For high-dimensional recordings with $N$ time series sharing a common category (e.g., a brain area), spectral PCA reduces the $N$-dimensional signal to a single canonical spectral mode $v(\omega)$. The area-level cross-spectrum is computed as the scalar bilinear projection $v_1^H \hat{S}_{12} v_2$, enabling area-level stitching.
(F) When the number of $N$ is large, graphical lasso (GLASSO) regularization imposes $\ell_1$ sparsity on the precision matrix $\hat{S}^{-1}(\omega)$, retaining only the strongest cross-spectral couplings and preventing overfitting in the subsequent PDC estimation.}
  \label{fig1}
\end{figure}

\subsection{Processes the method applies to}
The pipeline requires only a matrix-valued second-order spectrum $S(\omega)$ that is Hermitian, bounded above and below by positive constants, and admits a canonical minimum-phase factorization $S(\omega)=H(\omega)\,\Sigma\,H(\omega)^H$ (formal statement and proof in Appendix~\ref{S1_Appendix}). For continuous-valued processes $S(\omega)$ is the ordinary power spectral density (the Fourier transform of the autocovariance); for point processes it is the Bartlett spectrum, defined through the second-order intensity statistics. We describe three families of processes that satisfy this requirement.

Continuous-valued linear processes are the most familiar instance. Any zero-mean wide-sense stationary, purely non-deterministic process whose spectral density $S(\omega)$ is bounded away from zero and infinity (Eq~\eqref{eq:bdd-spectrum}) admits, by the Wold decomposition and the matrix Wiener--Hopf factorization, a unique minimum-phase factorization $S(\omega)=H(\omega)\,\Sigma\,H(\omega)^H$ with $H(\omega)$ and $H(\omega)^{-1}$ both analytic in the unit disk. Finite-order VAR processes, infinite-order VAR processes, vector moving-average (VMA) processes, and vector ARMA (VARMA) processes are all special cases: each corresponds to a particular rational or convergent form of $H(\omega)$, yet the directed information flow is recovered identically from the factorization without committing to any one of these parametric representations.

Point processes such as neural spike trains have no power spectral density in the ordinary sense: a train of discrete events has no continuous waveform whose autocovariance can be transformed. The role of $S(\omega)$ is instead played by the Bartlett spectrum, which describes the second-order structure of the firing intensity rather than of a sample path~\cite{bremaud2002power}. Intuitively, one passes from discrete spike times to the underlying firing rate, and the Bartlett spectrum captures how that rate co-varies across channels at each frequency. For the multivariate Hawkes process---a standard model in which each spike transiently raises the firing rate of connected units through excitation kernels---this Bartlett spectrum has a closed form that is again of the canonical type $S(\omega)=H(\omega)\,\Sigma\,H(\omega)^H$ whenever the process is stable (Appendix~\ref{S1_Appendix}). The directed structure is then carried entirely by the excitation kernels, so that an excitatory or inhibitory influence from one neuron onto another appears directly as directed information flow at the relevant frequencies. In practice the kernels never need to be known---the Bartlett spectrum is estimated directly from spike times by a multitaper point-process spectral estimator (detailed in the area-level analysis below), whose mean-rate correction removes the trivial contribution each train makes to its own spectrum.

The two cases combine naturally. Many neural recordings contain both spike trains and continuous signals such as local field potentials or calcium fluorescence, and the interactions of interest run in both directions---spike-to-field and field-to-spike. Because the pipeline depends only on a well-defined cross-spectral matrix, no separate machinery is required: point-process spectral theory defines valid auto- and cross-spectra between spike trains and continuous signals~\cite{jarvis2001sampling}, so the mixed spectral matrix is assembled from continuous--continuous, point--point (Bartlett), and point--continuous blocks and factorized exactly as before. A natural generative model here is a point process whose stochastic intensity follows an autoregressive form and is itself driven by the continuous signal; its spectrum can be obtained by estimating the intensity with any consistent smoother and computing the resulting cross-spectra. In practice this smoothing need not be a separate step: treating each spike train as a binary (zero--one) sequence and applying the multitaper estimator implicitly smooths the intensity, so the mixed cross-spectra follow directly from the raw events and continuous signals. Formal definitions of all three process families, together with the conditions under which the canonical factorization exists, are collected in Appendix~\ref{S1_Appendix} (\nameref{S_admissible}).

\subsection{The SPIDER estimator and its consistency}\label{sec:completeobs}

When every pair is co-observed (Condition~\ref{cond:complete}), the whole-network spectral matrix is reconstructed by averaging, for each pair $(i,j)$, the local cross-spectra from all blocks that contain both channels.
\begin{definition}[Stitched PSD under complete observation] \label{def:stitched}
For each $u \in [m]$, let $\hat{S}^{\Omega_u}(\omega)$ denote a non-parametric estimator of the power spectrum of $\mathbf{X}_{\Omega_u}$. For $i, j \in [K]$ and $\omega \in [-\pi, \pi)$, the \emph{stitched} power spectrum estimator is
\begin{equation} \label{eq:agrrespec}
\hat{S}_{ij}(\omega) \;=\; \frac{1}{\#\{u : i,j \in \Omega_u\}}\, \sum_{u : i,j \in \Omega_u} \hat{S}^{\Omega_u}_{ij}(\omega).
\end{equation}
\end{definition}
Each local cross-spectrum is estimated non-parametrically (multitaper spectral estimation; Appendix~\ref{S_mathdetails}). In high dimensions the stitched matrix need not be positive definite, so we regularize its inverse with a complex Hermitian graphical lasso (GLASSO), which guarantees a valid spectral factorization and reduces variance when the underlying dependence structure is sparse; the GLASSO formulation (and the matrix-completion step used below) are given in Appendix~\ref{S_mathdetails}. The directed information flow is then followed by the canonical (Wilson) factorization of the regularized matrix and the PDC normalization.

Under the completeness condition and a standard local spectral-consistency requirement on the per-block estimators (Assumption~\ref{ass:local}, Appendix~\ref{S_mathdetails}), SPIDER is uniformly consistent for every directed pair as the per-block sample sizes grow (Theorem~\ref{thm:complete_short}). Crucially, this holds even in genuinely high-dimensional regimes: the number of variables $K$ may grow nearly exponentially in the per-block sample size, provided each series has rapidly decaying temporal dependence~\cite{zhang2025spectral}---so the method applies where $K \gg T_u$ would defeat classical spectral theory.

\subsection{Recovering never-co-observed pairs}

When some pairs are never co-observed (Condition~\ref{cond:complete} fails), the corresponding spectral entries are missing, and they cannot simply be ignored: because PDC is a global functional, errors in unobserved entries propagate through inversion and factorization to the observed ones. We therefore complete the partially observed spectral matrix, frequency by frequency, using nuclear-norm minimization under the near-low-rank prior introduced above, before factorization (Appendix~\ref{S_mathdetails}). Recovery is then guaranteed conditionally: if completion is accurate, SPIDER remains uniformly consistent (Theorem~\ref{thm:incomplete}), which decouples the guarantee from the particular completion algorithm. We also give explicit, checkable conditions on the recording design under which nuclear-norm completion succeeds---essentially, that spectral energy is spread across many channel pairs and that enough pairs are co-observed, the required coverage scaling as $r K \log K$ up to logarithmic factors, where $r$ is the effective dimensionality (Proposition~\ref{prop:nnm-sufficient}, Appendix~\ref{S1_Appendix}). Because electrode placement varies unpredictably across sessions and animals, the set of co-observed pairs is well modeled as a random graph, and this coverage requirement is mild whenever $r$ is small relative to the number of areas.

\subsection{Area-level analysis for multi-session electrophysiology}\label{sec:areapdc}

In multi-session population recordings the individual neurons recorded within a brain area differ from session to session, whereas the \emph{areas} provide a stable set of variables. We therefore summarize each area by a single canonical population mode and stitch at the level of areas. For each session and area we estimate the neuron-by-neuron cross-spectral matrix with a multitaper point-process estimator---mean-rate corrected, so that each spike train's trivial self-contribution is removed---and reduce it to a one-dimensional area signal by projecting onto its leading frequency-domain principal component (spectral PCA; Fig~\ref{fig1}E). The resulting area-by-area cross-spectra are normalized to coherences, so that absolute power differences across sessions do not contaminate the pooled estimate. We retain an area only if its leading spectral mode is well defined, explaining at least a fixed fraction of the within-area spectral variance ($\rho_{\min}=0.10$ here); this assumes that the dominant population dynamics of an area (for example theta in CA1 or beta in motor cortex) are conserved across animals even though the specific neurons differ. The multitaper and spectral-PCA equations are given in Appendix~\ref{S_mathdetails}.

Stitching across sessions is valid only when local estimates of the same area-pair are reproducible across sessions---the empirical counterpart of the local spectral-consistency condition behind Theorem~\ref{thm:complete_short}. We therefore add a session-selection step: a session is retained only if the correlation between its area-pair spectral profiles and those of a reference session exceeds a data-driven threshold (Validation~2, \nameref{sec:neuropixel}). The session-averaged coherences are then assembled into the global matrix as in~\eqref{eq:agrrespec}, missing entries are imputed where needed, and the matrix is factorized with Wilson's algorithm to yield the area-level SPIDER estimate.

\subsection{Evaluation metrics}

Three complementary metrics assess estimation performance from different perspectives.

\paragraph{Mean squared error of PDC}
For quantitative assessment of estimation accuracy at the level of individual directed interactions of each frequency, we compute the mean squared error (MSE) of the off-diagonal entries of the PDC magnitude, averaged across all directed pairs.

\paragraph{Pearson correlation of integrated information flow}
MSE is a pointwise, frequency-resolved measure and may be dominated by high-PDC entries. To assess whether the overall pattern of directed information flow is faithfully recovered, we report the Pearson correlation between the integrated information-flow matrices derived from the SPIDER and reference estimates.
Directed information flow from $X_\ell$ to $X_k$ is quantified through \eqref{eq:infoflow}.
The Pearson correlation is computed between the vectorized off-diagonal entries of the estimated and reference
information flow matrices, providing a scale-invariant measure of structural similarity that is robust to overall amplitude differences between estimation conditions.

\paragraph{Area under the ROC curve}
To evaluate the recovery of directed network topology, independent of edge strength, we compute the area under the receiver operating characteristic curve (ROC-AUC). This metric is reported where a binary reference is available, in particular for the comparison of the SPIDER information flow against anatomical connectivity (Fig~\ref{fig7}H): we sweep the threshold on the estimated information flow, compute true- and false-positive rates against the reference adjacency matrix, and integrate. An AUC of 0.5 corresponds to chance and 1.0 to perfect edge recovery.

\subsection{Implementation and computational complexity}

Both the GLASSO and the nuclear-norm completion operate on a real-symmetric augmentation of the complex spectral matrix (Appendix~\ref{S_mathdetails}); the GLASSO problem is solved by block coordinate descent~\cite{friedman2008sparse}. In finite samples we fix a small penalty ($\lambda=0.01$), trading a little bias for reduced variance; in practice $\lambda$ can be chosen by cross-validation on held-out co-observed pairs. The cost is dominated by frequency-wise matrix operations---GLASSO, the completion singular value decomposition, and Wilson factorization are each $O(K^3)$ per frequency---giving an overall complexity of $O(T\cdot K^3)$ on a frequency grid of size $\asymp T$. Because the frequencies are processed independently, the pipeline is embarrassingly parallel, which is the natural axis of parallelization in our implementation. The complete procedure is summarized as Algorithm~\ref{alg:stitched_pdc} in Appendix~\ref{S_mathdetails}.


\section{Results}

We first establish correctness and characterize the behavior of SPIDER on controlled simulations with known ground truth, then demonstrate its intended use on three complementary large-scale neural datasets: two-photon calcium imaging of mouse visual cortex, brain-wide Neuropixels electrophysiology pooled across sessions, animals, and laboratories, and human intracranial EEG pooled across patients.

\subsection{SPIDER recovers PDC under complete pairwise observation}

We first verify the method on a small system where the ground truth is known analytically. Time series are simulated from a stationary VAR(1) process (Fig~\ref{fig2}A). We model exactly the situation in which $X_3$ drives both $X_1$ and $X_2$ but $X_1$ and $X_2$ have no direct interaction. All pairs $(i,j) \in [K] \times [K]$ are observed in independent blocks, satisfying Condition~\ref{cond:complete}.
We compare four estimates: the theoretical PDC obtained from the Fourier transform of the VAR(1) coefficients, the non-parametric PDC computed from a single fully simultaneous recording via~\eqref{eq:nonparamPDC}, the SPIDER estimate obtained from the asynchronous blocks via Algorithm~\ref{alg:stitched_pdc}, and a pairwise-only control that estimates local PDC without integrating information across blocks.
The SPIDER estimate tracks the simultaneous estimate almost exactly across all pairs and frequencies, and both closely approximate the theoretical PDC (Fig~\ref{fig2}B). The pairwise-only control deviates systematically, especially between $X_1$ and $X_2$, where interaction does not exist in reality, confirming that aggregating the full spectral matrix across blocks is necessary to recover the correct directed structure.
Together, this confirms that asynchronous block recordings, when each pair is jointly observed at least once, contain the same directional information as a single simultaneous recording.

The VAR framework underlying the theoretical guarantees assumes linear dynamics, yet many real systems of interest, including neural circuits, are inherently nonlinear. We therefore test whether the stitching pipeline generalizes to a biophysically realistic network of Izhikevich spiking neurons, which exhibit nonlinear spike generation. The network consists of two excitatory neurons $N_1(\text{E})$ and $N_2(\text{E})$ and one inhibitory neuron $N_3(\text{I})$, connected with known directed synaptic weights (Fig~\ref{fig2}C; see Appendix~\ref{S1_Text} for model details). Because spike trains are point processes, we apply the multi-taper point-process spectral estimator of~\eqref{eq:tapered_ft}~\cite{mitra1999analysis}.
The SPIDER estimates closely track the full-simultaneous reference across all nine entries of the $3\times3$ PDC matrix and across the frequencies (Fig~\ref{fig2}D), demonstrating that the pipeline handles the spike-timing dynamics without degradation. Because the estimate is recovered at every frequency (Fig~\ref{fig2}D), directed interactions can be localized to specific frequency bands; the integrated information-flow summaries used for the larger comparisons below collapse this axis for compactness. To assess biological validity, we compare the estimated directed information flow to a ground-truth effective influence matrix derived from the simulated network. The matrix is computed as synaptic coupling strength multiplied by presynaptic firing rate, providing a proxy for the expected rate of information transmission along each connection (Fig~\ref{fig2}E). Both the full-simultaneous estimate (Pearson $r=0.89$) and the SPIDER estimate ($r = 0.96$) recover the synaptic organization faithfully. Together, these results demonstrate that SPIDER generalizes beyond linear Gaussian processes to the nonlinear, point-process setting that characterizes real neural recordings.

\begin{figure}[h!]
  \centering
  \includegraphics[width=\columnwidth]{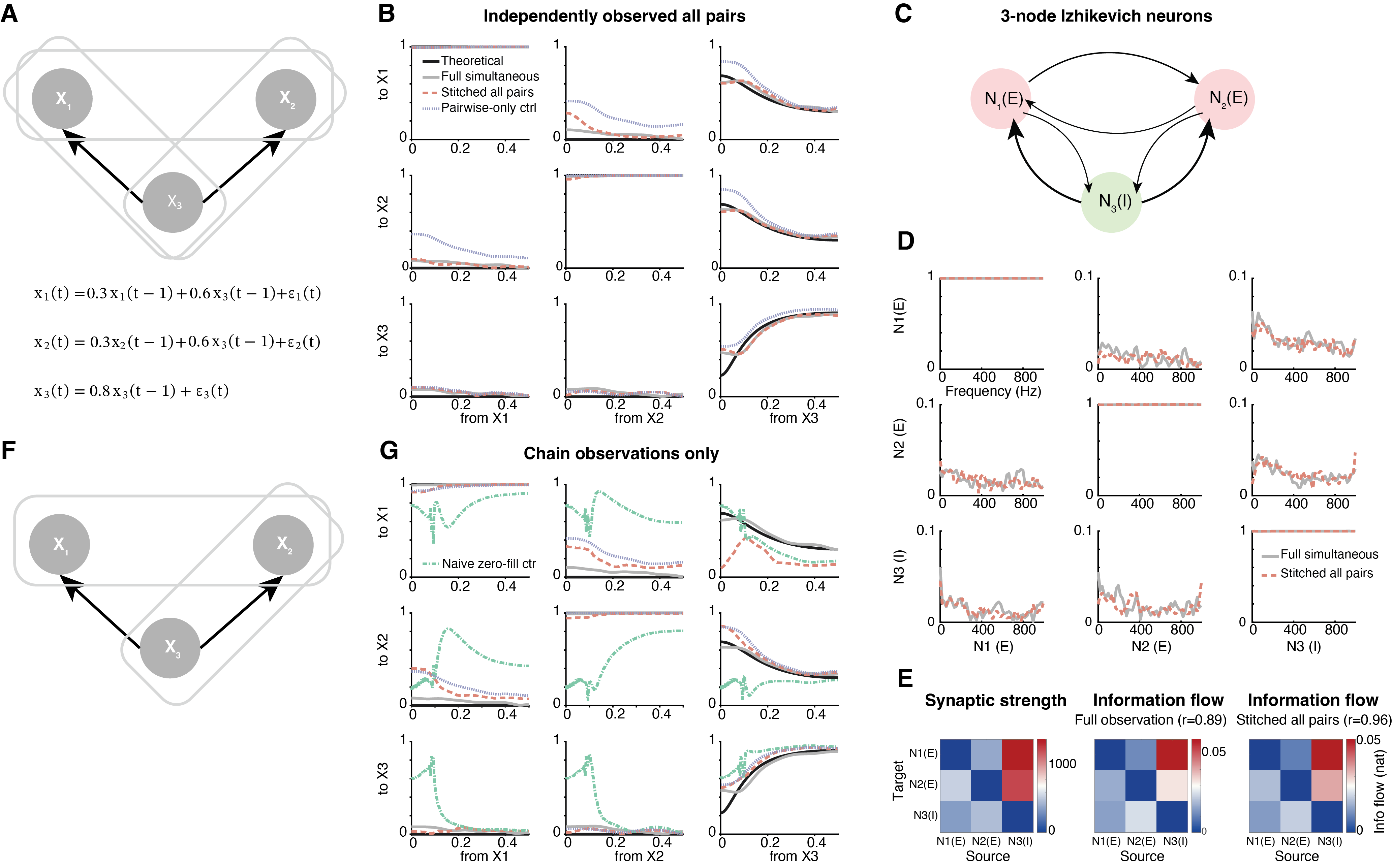}
  \vspace{0.5cm}
  \caption{{\bf SPIDER recovers directed information flow under complete and incomplete observation and in biophysically realistic spiking networks.} (A) Three-node VAR(1) system used for validation. $X_3$ drives both $X_1$ and $X_2$, while $X_1$ and $X_2$ have no direct interaction. All three pairwise subsets are observed in independent asynchronous blocks.
  (B) PDC estimates (rows: target node; columns: source node) under complete pairwise observation.
  (C) Three-node biophysically realistic network of Izhikevich neurons: two excitatory populations $N_1(E)$ and $N_2(E)$, and one inhibitory population $N_3(I)$.
  (D) Frequency-resolved PDC for the SPIDER estimate (red dashed) and full-simultaneous reference (gray).
  (E) Correspondence between ground-truth synaptic strength (left; synaptic weight times firing rate) and integrated directed information flow for full observation (middle; Pearson $r=0.89$ with synaptic strength) and SPIDER all-pairs estimation (right; $r=0.96$).
  (F) Incomplete observation design. Only the chain-structured pairs $(X_1, X_2)$ and $(X_2, X_3)$ are jointly observed; $(X_1, X_3)$ is never co-observed.
  (G) PDC estimates under chain-only observation.
  }
  \label{fig2}
\end{figure}

\subsection{Matrix completion enables recovery under incomplete pairwise observation}

We next consider the case where Condition~\ref{cond:complete} fails. Using the same three-node VAR(1) system, we adopt a chain observation design in which only the pairs $(X_1, X_2)$ and $(X_2, X_3)$ are jointly observed, while $(X_1, X_3)$ is never co-observed (Fig~\ref{fig2}F).
Recall that the true structure has $X_3$ driving both $X_1$ and $X_2$ with no direct $X_1$--$X_2$ interaction; because $X_3$ is observed with $X_2$ but not with $X_1$, any local estimate restricted to the $(X_1, X_2)$ block will conflate the shared driver with a spurious direct connection.
SPIDER estimates obtained via NNM recover the directed structure more accurately than the local pairwise control across all observed directions, and in particular substantially mitigate the spurious $X_1$--$X_2$ interaction that contaminates local estimates (Fig~\ref{fig2}G). The residual spurious component is not fully eliminated, consistent with the fact that the $(X_1, X_3)$ cross-spectrum must be inferred.
By contrast, a naive baseline that fills missing entries of $\hat{S}(\omega)$ with zeros fails to reproduce even some of the observed entries (e.g., $2 \leftarrow 3$), illustrating that ignoring missingness is not benign: errors in unobserved entries propagate to observed ones through the inversion and factorization steps.

\subsection{High-dimensional performance and the role of completion}

To assess performance in higher dimensions, we simulate $K=50$-dimensional stationary VAR(1) processes with randomly generated full-rank coefficient matrices and threshold to achieve target sparsity levels of 5\%, 50\%, 75\%, and 95\% zeros.
Innovations are i.i.d.\ Gaussian with covariance $0.1\,\mathbf{I}_K$. Each simulation runs for $n = 10{,}000$ time steps after discarding $1{,}000$ burn-in samples.
Observation blocks are constructed as contiguous, partially overlapping subsets of the $K$ variables with a fixed overlap ratio of 50\% between adjacent blocks; varying the block size ratio across \{0.1, 0.2, 0.4, 0.6\} yields the four missingness levels of 30\%, 48\%, 72\%, and 86\% shown in the Fig~\ref{fig4}.
All results are averaged over 10 independent replicates. For each configuration, we compute the MSE of the off-diagonal entries of the SPIDER estimate against two references: the theoretical PDC computed from the known VAR(1) coefficients (Fig~\ref{fig4}A), and the non-parametric PDC from a single fully simultaneous recording, which serves as a finite-sample oracle upper bound (Fig~\ref{fig4}B).

We compare seven conditions (Fig~\ref{fig4}A-B): SPIDER estimates with NNM completion for all off-diagonal pairs, observed pairs, and never-observed pairs; a local baseline that estimates PDC only within each observation block, restricted to co-observed pairs and conditioned solely on observed variables; and the corresponding three subsets of the naive zero-fill baseline. The local baseline directly answers the question of what ignoring asynchrony would obtain: each block is restricted to pairs within $\Omega_u$ and conditioned only on $\Omega_u \setminus \{i,j\}$, not on the full $[K] \setminus \{i,j\}$. As discussed above, this is a fundamentally different estimand from the whole-network PDC, and per-block estimates are subject to confounding by unobserved variables.

SPIDER with NNM completion produces significantly lower MSE than the naive zero-fill baseline for both references ($p < 0.05$, paired t-test), and at low to moderate missingness ($\leq$72\%), it outperforms the local per-block baseline for observed pairs ($p < 0.05$, paired t-test; Fig~\ref{fig4}A\&B). At very high missingness, the advantage of the NNM stitching reduces, especially for unobserved pairs: when most entries are unobserved, the structural prior cannot compensate for the lack of information, and completion becomes unreliable. The transition from the regime in which completion meaningfully helps to the regime in which it does not is consistent with the coverage condition of Proposition~\ref{prop:nnm-sufficient}: as $|\mathcal{O}|/K^2$ falls below $r \log K \cdot \log T / K$, the nuclear-norm program loses its consistency guarantees. The local per-block baseline becomes more advantageous if missingness is high: at 86\% missingness, it outperforms SPIDER.

\begin{figure}[h!]
  \centering
  \includegraphics[width=\columnwidth]{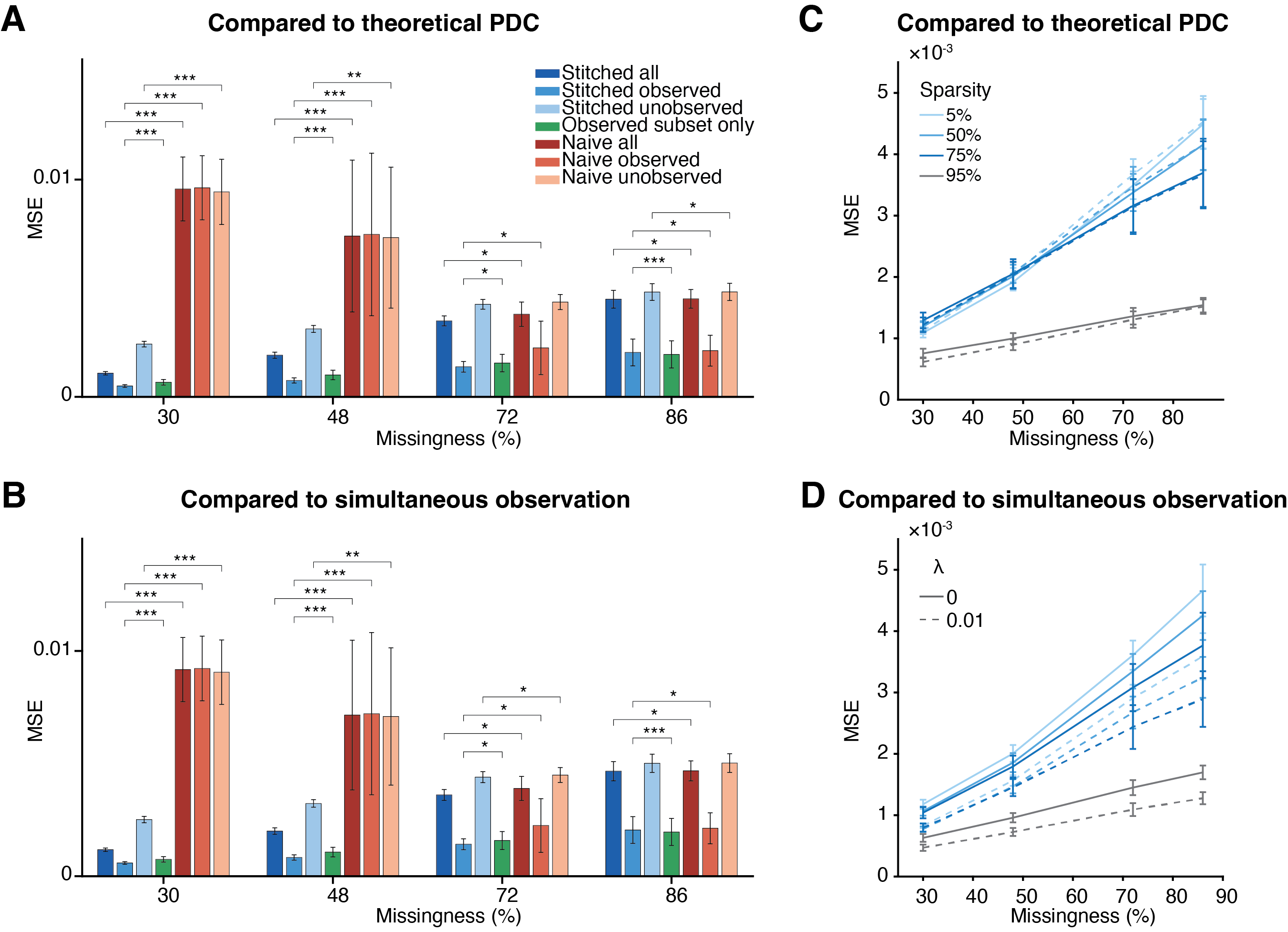}
  \vspace{0.5cm}
  \caption{{\bf SPIDER accuracy in high dimension.} Simulations use $K = 50$-dimensional stationary VAR(1) processes (10 replicates). MSE of off-diagonal PDC entries is reported against (A) the theoretical PDC and (B) the PDC from a single fully simultaneous recording. Error bars are standard deviations across simulation replicates; asterisks denote pairwise significance from paired t-tests ($^{*}p < 0.05$, $^{**}p < 0.01$, $^{***}p < 0.001$). (C, D) MSE of off-diagonal SPIDER PDC estimates with (dotted lines) and without (solid lines) GLASSO regularization, across varying levels of VAR coefficient sparsity (color) and missingness (horizontal axis). Error bars are standard deviations.}
  \label{fig4}
\end{figure}

\subsection{Effect of GLASSO regularization}

We next isolate the contribution of the GLASSO regularization step by varying both the regularization strength $\lambda$ and the sparsity of the underlying VAR coefficient matrix. When using the theoretical PDC as reference, the benefit of GLASSO regularization ($\lambda \in \{0,\, 0.01\}$) depends jointly on missingness and coefficient-matrix sparsity (Fig~\ref{fig4}C). When the underlying connectivity is sparse, GLASSO regularization reduces the error of the SPIDER estimate, with the largest gains at low missingness. As sparsity decreases, the benefit of regularization shrinks, and at very dense connectivity GLASSO can hurt accuracy.
When compared against the simultaneous-recording reference, the regularization choice $\lambda = 0.01$ outperforms the unregularized estimator across all sparsity levels, indicating that even when the true coefficient matrix is dense, the $\ell_1$ penalty corrects for finite-sample variance introduced by the stitching and completion steps (Fig~\ref{fig4}D). MSE increases monotonically with missingness under all conditions.

\subsection{Application to calcium imaging}

We applied the method to two-photon calcium imaging data from mouse visual cortex in the Allen Brain Observatory~\cite{de2020large}. We analyzed a 600~s spontaneous-activity segment ($\approx$ 18{,}000 samples at 30~Hz) during which 227 neurons were recorded simultaneously. To ensure reliable information-flow estimation, we restricted the analysis to neurons exhibiting sufficient activity (dF/F signal exceeding $2.5$ standard deviations for at least 200 time points), yielding 59 neurons. PDC here is computed on the dF/F signal rather than on deconvolved spike rates, so the recovered interactions reflect dynamics at the timescale resolved by the calcium indicator and should not be over-interpreted at faster timescales.

To assess estimation stability under finite samples, we examined how non-parametric PDC estimates depended on recording duration. Treating the full 600~s segment as a reference, we computed the mean squared error between PDC magnitudes estimated from shorter segments and the full-data estimate, excluding diagonal entries. Error decreased systematically with increasing data length, with a pronounced drop around 200~s (approximately one-third of the recording), beyond which improvements were gradual (Fig~\ref{fig6}A). The Pearson correlation between integrated information-flow matrices~\eqref{eq:infoflow} estimated from truncated and full data increased monotonically with sample size, indicating convergence. These results suggest that stable PDC estimation is achievable at the data lengths considered.

We then divided the recording into three sequential time blocks with 50\% neuron overlap (overall missingness $37\%$) and applied SPIDER with NNM completion. The reconstructed information flow matrix recovered many of the strongest directed interactions observed under full simultaneous recording (Fig~\ref{fig6}B,C). Observed entries showed strong agreement with the simultaneous-recording reference ($r = 0.86$). Imputed entries also exhibited a significant positive correlation with the reference ($r = 0.26$, $p = 8.9 \times 10^{-22}$; Fig~\ref{fig6}D), demonstrating partial recovery of unobserved interactions consistent with the simulation results above.

We examined the effect of GLASSO regularization on reconstruction accuracy. Moderate regularization significantly reduced error ($F = 10.37$, $p = 1.8 \times 10^{-7}$), while overly strong regularization degraded performance through excessive sparsification (Fig~\ref{fig6}E). Regularization reduced error for observed interactions but increased error for imputed ones, consistent with over-attenuation of weak inferred connections. Despite these effects on error magnitude, the overall information-flow pattern was largely preserved: Pearson correlations with the simultaneous-recording reference remained stable across $\lambda$ (Fig~\ref{fig6}F). The principal difference was between observed and imputed entries, with imputed entries showing systematically lower correlations ($p = 4.2 \times 10^{-60}$) but remaining significantly above zero for all $\lambda$ tested ($p < 0.05$).

Together, these results show that SPIDER robustly recovers directed information flow from partial and asynchronous neural recordings, with regularization primarily improving stability without altering the dominant directional structure.

\begin{figure}[h!]
  \centering
  \includegraphics[width=\linewidth]{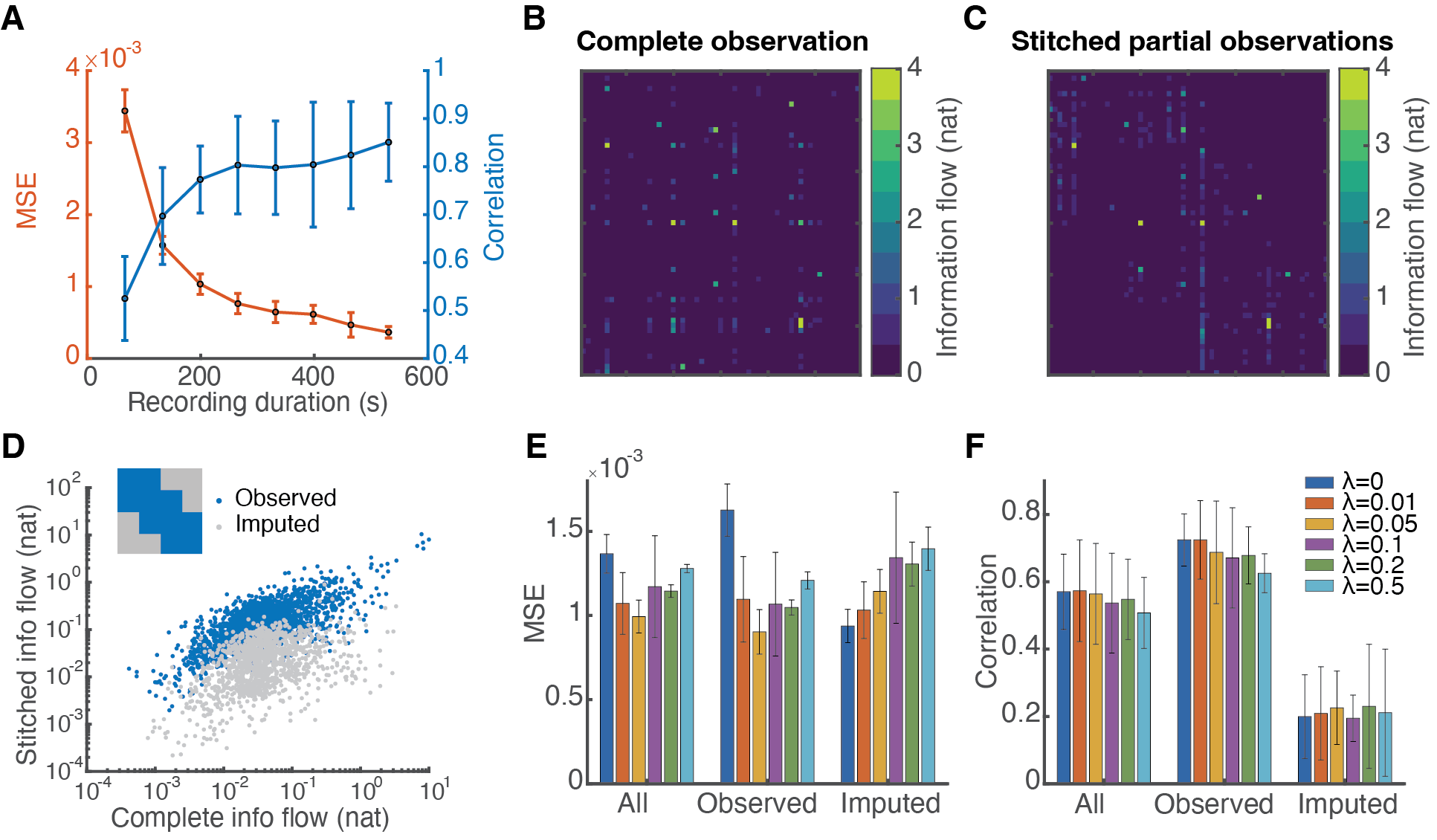}
  \vspace{0.5cm}
  \caption{{\bf Directed information-flow reconstruction from mouse calcium imaging using SPIDER.} (A) Estimation stability as a function of data length, quantified by mean squared error of PDC magnitudes and correlation of integrated information-flow matrices with the full-length estimates. Truncated datasets were generated via random circular shifts; points are medians and error bars indicate robust variability ($1.4826 \times \mathrm{MAD}$). (B) Information-flow matrix (diagonal zeroed) from complete simultaneous observation of 59 neurons. (C) Information-flow matrix reconstructed from three partially overlapping time blocks using SPIDER (multitaper: $n_{\mathrm{win}} = 128$, $[NW, L] = [5, 9]$; GLASSO $\lambda = 0.01$). (D) Element-wise comparison between SPIDER and simultaneous-recording information-flow matrices; inset distinguishes observed from imputed entries. (E) Effect of GLASSO regularization on the mean squared error of PDC magnitudes for all, observed, and imputed pairs. (F) Effect of GLASSO regularization on the correlation between SPIDER and simultaneous-recording information-flow matrices. (E--F) are averaged over 15 random block assignments; error bars denote standard deviations.}
  \label{fig6}
\end{figure}

\subsection{Application to Neuropixels data}\label{sec:neuropixel}

Neuropixels probes enable simultaneous recording of hundreds of single neurons in multiple brain areas, including deep subcortical structures. When aligned to a standard brain atlas, the shared area annotations across animals provide the anchor for stitching directed information flow across sessions, animals, and laboratories. We applied our method to the International Brain Laboratory (IBL) brain-wide map dataset~\cite{international2025reproducibility,international2025brain}, comprising Neuropixels recordings collected across 12 laboratories under a standardized experimental pipeline targeting 279 brain areas in 139 adult mice. To avoid stimulus-locked responses, we restricted the analysis to passive and spontaneous activity epochs, which lasted approximately 600 seconds per session.

Recordings were subject to a multi-stage quality control pipeline. We retained only neurons classified as good isolated units by the automated spike-sorting quality metrics, restricted to sessions with at least two valid brain areas simultaneously recorded, and excluded fiber tract regions.
At the neuron level, units with a mean firing rate below 0.2 Hz were discarded to ensure sufficient spike counts for stable spectral estimation.
At the area level, we required at least 5 simultaneously recorded neurons per area, and further excluded areas whose spectral PC1 explained less than 10\% of the total spectral variance, which we used as a data-driven threshold indicating insufficient shared population-level structure.

The area-level SPIDER procedure (stitching via matched brain areas) is detailed in Methods and illustrated in Fig~\ref{fig7}A. Since no ground-truth directed connectivity exists for the full set of brain areas, we validated the approach in three complementary ways.

\paragraph{Validation 1: within-session split}
Within the same session, the simultaneously recorded full set of areas provides an upper bound on estimation accuracy. To assess how closely stitching approximates this bound, we artificially partitioned each session into two halves, each covering a different but partially overlapping subset of brain areas, and then compared the SPIDER estimate against the PDC estimated from the full simultaneous observation (Fig~\ref{fig7}B).
In two sessions with 20 randomized brain area assignments each, the overall agreement between the SPIDER and fully-observed PDC was high ($r = 0.72 \pm 0.29$, mean $\pm$ SD; Fig~\ref{fig7}C).
Area pairs that were co-observed showed strong agreement ($r = 0.88 \pm 0.12$) and pairs requiring imputation were weaker ($r = 0.48 \pm 0.35$). This result confirms that our area-level SPIDER approach is feasible.

\paragraph{Validation 2: local spectral consistency predicts stitching accuracy}
Theorem~\ref{thm:complete_short} guarantees convergence of the SPIDER estimate provided that local spectral consistency holds across sessions. To test whether this condition is met empirically, and to derive a principled session-selection threshold, we identified 20 target sessions with a relatively large number of simultaneously recorded areas ($n=5.7\pm0.9$) and, for each, designated two source sessions sharing a partially overlapping subset of areas (Fig~\ref{fig7}D).
We used the area-level cross-spectral density (CSD) and PDC of the target session as reference, and measured (i) the Pearson correlation between the source and target CSDs and (ii) the correlation between the SPIDER information flow derived from the source sessions and the target PDC. Across all 20 sessions, higher source-target CSD similarity was associated with higher SPIDER information-flow accuracy (Spearman correlation $\rho = 0.71$, $p = 7.1\times10^{-4}$; Fig~\ref{fig7}E), indicating that local spectral consistency is a strong predictor of stitching accuracy.
From this relationship, we determined the session-selection threshold as the CSD-correlation value at the zero-crossing of the lower 95\% confidence band of the fitted regression line for the observed pairs (Fig~\ref{fig7}E).

\paragraph{Validation 3: correspondence with anatomical connectivity}
Applying the threshold determined in Validation~2, we identified the largest stitchable graph whose missingness remained below the limit ($\sim$70\%) identified in simulation (Fig~\ref{fig4}), beyond which stitching accuracy becomes indistinguishable from that of null models.
Maximizing area coverage while keeping missingness below this bound, the final dataset retained 43 sessions from 12 laboratories, encompassing 50 brain areas and 4{,}074 neurons (Fig~\ref{fig7}F \& G; see Appendix~\ref{S2_Text} for session and brain area information). The selected sessions spanned cortical, hippocampal, thalamic, and midbrain regions, providing broad coverage of the mouse brain's major functional subsystems.

As an indirect validation, we compared the SPIDER directed information flow to structural connectivity from the Allen Mouse Brain Connectivity Atlas~\cite{oh2014mesoscale}, which provides whole-brain anterograde axonal tracing data. For each pair of areas in the stitched graph, we quantified directed anatomical connectivity as the maximum anterograde projection density from the source area to the target area. Sweeping the binarization threshold of the Allen projection matrix and computing the ROC curve against the SPIDER estimate yielded AUC values approaching 0.9 at high thresholds, indicating that the most prominent anatomical projections were reflected in the estimated information flow (Fig~\ref{fig7}H). The AUC for observed area pairs was consistently above chance, whereas unobserved pairs showed near-chance AUC. The anatomical correspondence is therefore carried by the directly co-observed pairs; directed edges between areas that were never co-recorded are recovered only through completion and should be read with corresponding caution.

Taken together, these three complementary validations establish that independently recorded Neuropixels datasets can be integrated across sessions, animals, and laboratories to recover directed information flow among brain areas that were never simultaneously recorded.

\begin{figure}[h!]
  \centering
  \includegraphics[width=\linewidth]{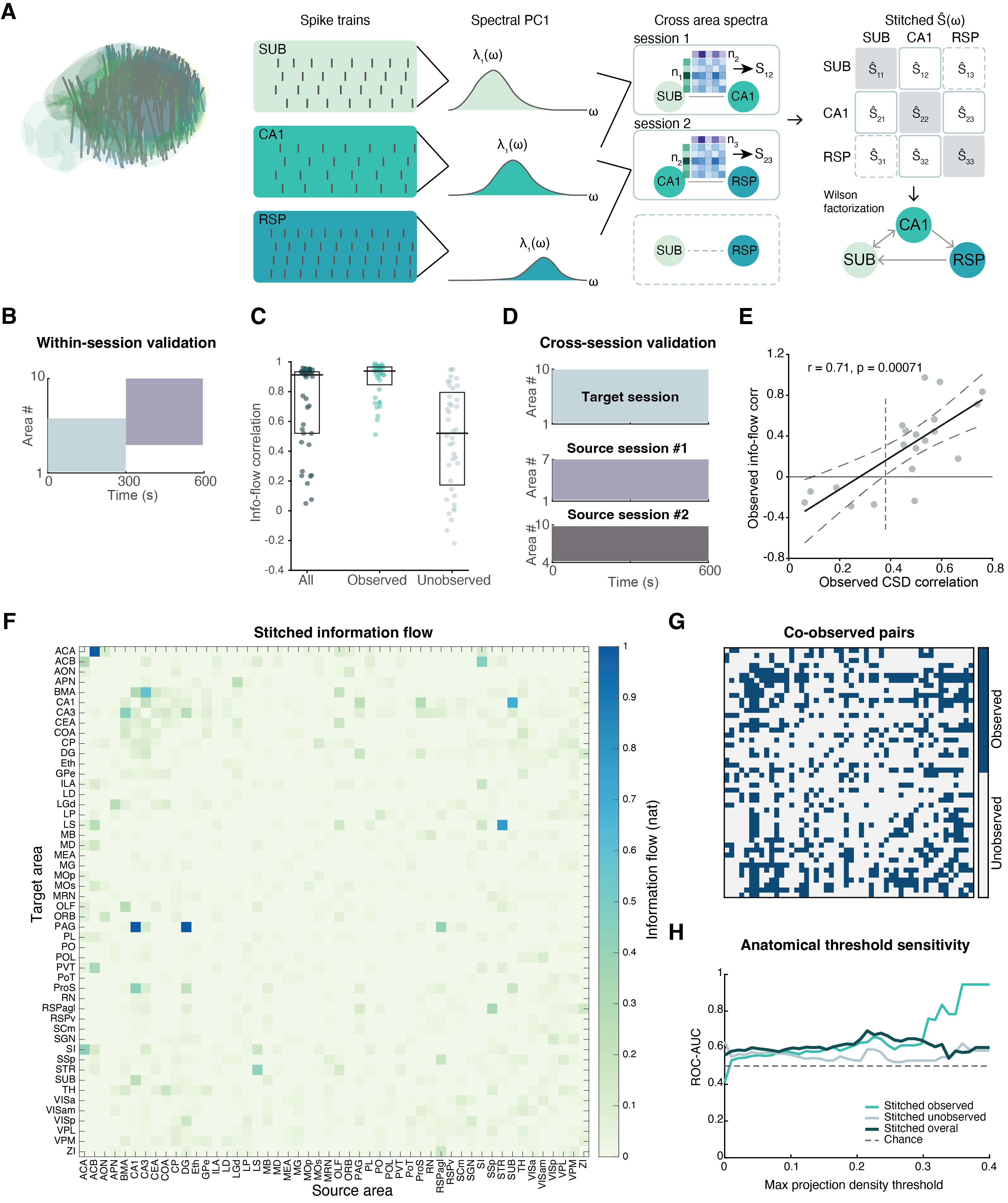}
  \vspace{0.0001cm}
  \caption{{\bf Directed information-flow reconstruction from IBL Neuropixels data.}
  (A) Pipeline for area-level stitching from Neuropixels data. Spike trains from each brain area are reduced to a single canonical spectral mode via spectral PCA. Local cross-area spectra estimated from two partially overlapping sessions are assembled into the stitched spectral matrix and factorized via Wilson's algorithm to calculate PDC.
  (B) Within-session validation design. Each session is split into two temporally non-overlapping halves covering different but partially overlapping subsets of areas.
  (C) Information-flow correlation between SPIDER and full-simultaneous estimates for all, observed, and unobserved area pairs.
  (D) Cross-session validation design. One target session with many simultaneously recorded areas is paired with two source sessions sharing overlapping but incomplete area subsets.
  (E) Observed information-flow correlation versus observed cross-spectral density (CSD) correlation across 20 target--source session triplets. The vertical dashed line marks the CSD-correlation threshold used for session selection.
  (F) SPIDER information-flow matrix across 50 brain areas. Color encodes integrated directed information flow in nats; diagonal is zeroed.
  (G) Co-observation map.
  (H) ROC-AUC as a function of the anatomical projection density threshold used to binarize the Allen Mouse Brain Connectivity Atlas.}
  \label{fig7}
\end{figure}

\subsection{A theta-band feedforward hierarchy of spontaneous information flow}\label{sec:hierarchy}

The validations above establish that SPIDER recovers directed information flow faithfully; we next asked whether the recovered brain-wide network carries interpretable organization that no single session could reveal. A directed network is \emph{feedforward}, or hierarchical, to the extent that its edges can be arranged so that influence flows consistently from lower to higher levels, and \emph{recurrent} to the extent that flow circulates in loops or is reciprocated. We quantified this with the trophic-incoherence framework of MacKay, Johnson and Sansom~\cite{mackay2020directed,johnson2014trophic}, which assigns each area a continuous \emph{trophic level} from the directed flow and measures the departure from a perfect feedforward hierarchy by an incoherence index $F_0\in[0,1]$ ($0$, perfectly feedforward; $1$, no hierarchy; Fig~\ref{fig:hierarchy}A and Appendix~\ref{app:hierarchy}). Significance was assessed against a direction-randomized null that preserves every undirected interaction strength but scrambles its direction.

Integrated over all frequencies, the SPIDER network was essentially non-hierarchical ($F_0\approx0.92$, indistinguishable from the direction-shuffled null, $p\approx0.10$): brain-wide spontaneous flow is, on the whole, recurrent. Resolving the estimate by frequency revealed a sharply localized exception. Trophic incoherence dropped to a single narrow minimum in the theta band (4--8~Hz; $F_0=0.72$, $z=-2.4$, $p=0.001$ against the null, surviving Bonferroni correction across five bands); a continuous frequency sweep confirmed a single minimum localized to the theta range (Fig~\ref{fig:hierarchy}B). The theta ordering was robust---leave-one-area-out jackknifing left it essentially unchanged (Spearman $\rho=0.99$)---and incoherence rose monotonically with frequency, so that by the gamma range the network was statistically indistinguishable from a directionless graph.

The theta hierarchy had a consistent anatomical polarity (Fig~\ref{fig:hierarchy}C): the hippocampal formation (subiculum, dentate gyrus, CA1) and ventral striatum (nucleus accumbens) occupied the source end, while midline and mediodorsal thalamus, periaqueductal gray, and medial prefrontal cortex (infralimbic area) formed the sinks. Decomposing the total directed flow into a feedforward gradient, a circulating (loop) component, and a reciprocal (mutual) component made the frequency dependence explicit (Fig~\ref{fig:hierarchy}D): reciprocal flow dominated in every band and grew from roughly one half in theta to $84\%$ in high gamma, whereas the feedforward share peaked in theta---about half of the total flow, with $92\%$ of the theta net flow directed up the hierarchy---and collapsed toward gamma.

Two points deserve emphasis. First, the finding is intrinsic to the data and required no anatomical prior: the hierarchy is read out from the directed spectra themselves. Second, it is a frequency-specific directional effect---feedforward in theta, recurrent in gamma---of exactly the kind that is invisible to time-domain or scalar connectivity measures and that emerges only because SPIDER recovers whole-network, frequency-resolved direction from data that were never recorded together. That theta is the band in which spontaneous brain-wide flow organizes into a hierarchy, with the hippocampal formation at its apex, is consistent with the established role of theta in coordinating hippocampal--cortical communication~\cite{buzsaki2002theta,colgin2013mechanisms}. We report this as a hypothesis-generating result: it rests on passive and spontaneous epochs and on directly co-observed area pairs, and it is an effective-connectivity statement rather than a claim about monosynaptic anatomy.

\begin{figure}[h!]
  \centering
  \includegraphics[width=\linewidth]{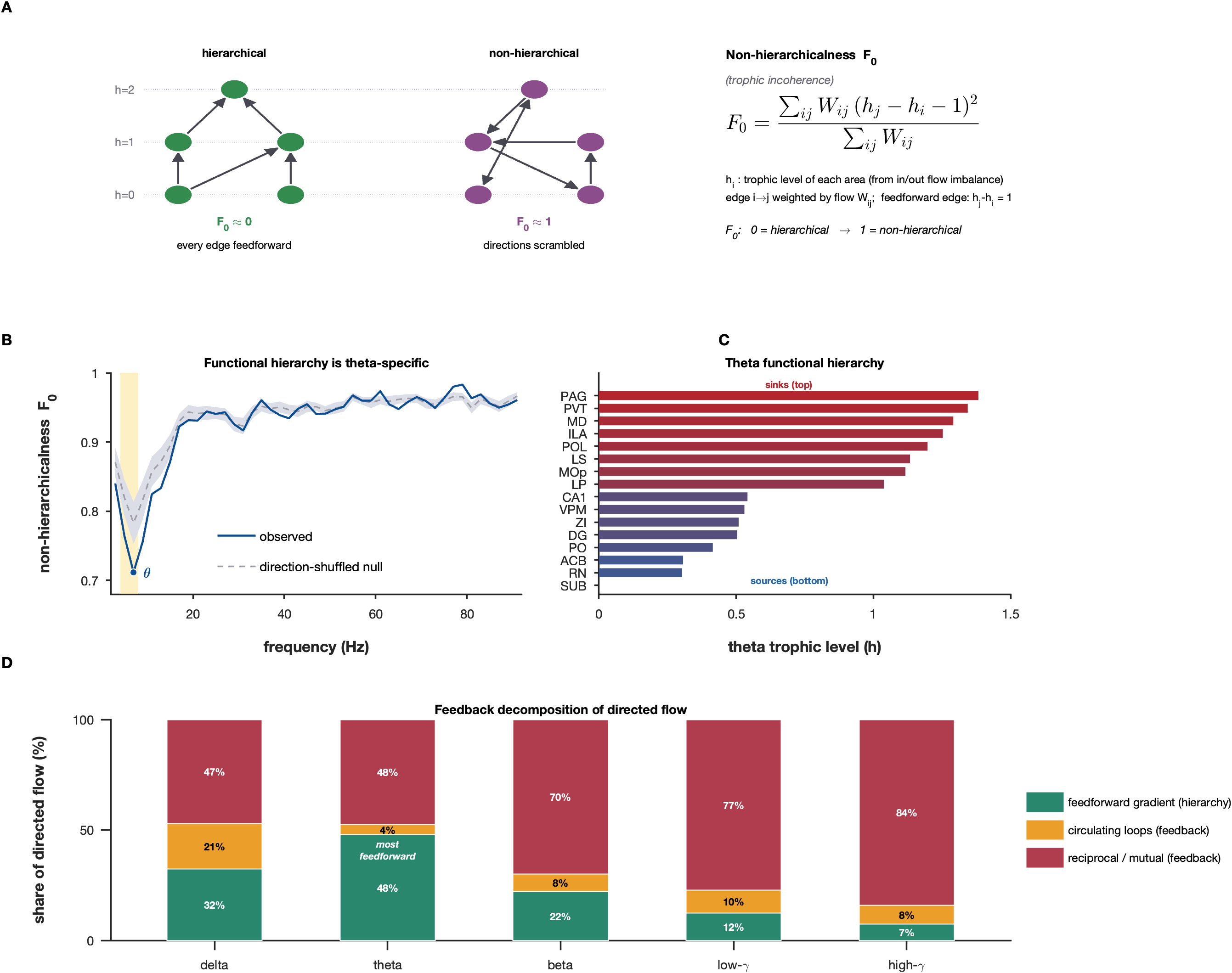}
  \caption{{\bf SPIDER reveals a theta-band feedforward hierarchy in spontaneous brain-wide flow.}
  (A) How hierarchy is quantified. A directed network is summarized by a trophic-incoherence index $F_0$ (its ``non-hierarchicalness''): each area receives a trophic level $h$, and $F_0$ is the flow-weighted variance of the level difference across edges, equal to $0$ for a perfectly feedforward hierarchy and approaching $1$ when edge directions are unstructured. Toy hierarchical and non-hierarchical networks are shown for reference.
  (B) Trophic incoherence of the SPIDER directed network as a function of frequency (sliding 4~Hz window; observed estimate in blue) against a direction-shuffled null (grey band, mean $\pm$1\,s.d.). The network departs from the null only in a narrow theta minimum (shaded) and is indistinguishable from directionless at higher frequencies.
  (C) Theta-band trophic levels of the most extreme areas (Allen Mouse Brain Atlas acronyms). Sinks (top of the hierarchy, red): PAG, periaqueductal gray; PVT, paraventricular, and MD, mediodorsal thalamus; ILA, infralimbic (medial prefrontal) cortex; POL, posterior limiting, and LP, lateral posterior thalamus; LS, lateral septum; MOp, primary motor cortex. Sources (bottom, blue): the hippocampal formation---SUB, subiculum; DG, dentate gyrus; CA1---together with ACB, nucleus accumbens (ventral striatum); VPM, ventral posteromedial, and PO, posterior thalamus; ZI, zona incerta; RN, red nucleus.
  (D) Decomposition of the total directed flow per band into a feedforward gradient (hierarchy), circulating loops, and reciprocal (mutual) flow. The feedforward share peaks in theta and is replaced by reciprocal flow toward gamma.}
  \label{fig:hierarchy}
\end{figure}

\subsection{Application to human intracranial EEG}\label{sec:ecog}

To test whether SPIDER transfers to human recordings, we applied it to the open multimodal iEEG dataset of Berezutskaya et al.~\cite{berezutskaya2022open} (OpenNeuro \texttt{ds003688}), comprising intracranial EEG from 51 patients implanted with subdural grids and stereo-EEG depth electrodes for presurgical monitoring. Here the stitching anchor is anatomical rather than experimental: each patient's electrodes were normalized to MNI space and labeled with the automated anatomical labeling (AAL) atlas~\cite{tzourio2002automated}, so that shared atlas regions play the role that shared brain areas played for the Neuropixels data, and patients play the role of sessions. As in the rodent analysis, we restricted attention to resting, spontaneous recordings. After retaining good-quality SEEG/ECoG channels and re-referencing to nearest-neighbor bipolar derivations within each electrode shaft, the bipolar cross-spectra were estimated by continuous multitaper---in place of the point-process estimator used for spike trains---each region was reduced to a single canonical spectral mode by spectral PCA (exactly as areas are reduced in the Neuropixels analysis), and the resulting region-by-region cross-spectra were passed unchanged to the stitching, completion, and Wilson-factorization steps. Because electrode coverage in this cohort is strongly left-lateralized, we report the left hemisphere. The final stitched graph comprised 38 left-hemisphere cortical AAL regions pooled from 43 patients, none recorded together, of which $498/703$ ($71\%$) region pairs were directly co-observed in at least one patient; the remaining pairs were recovered by completion, a missingness ($29\%$) well within the regime in which stitching is reliable (Fig~\ref{fig4}).

SPIDER returned a directed whole-cortex information-flow estimate across the 38 regions (Fig~\ref{fig:ecog}A), recovered entirely from patients that were never recorded together. Applying the same trophic-incoherence analysis used for the rodent network, the directed hierarchy was strongly concentrated in the theta band: trophic incoherence was low in theta ($F_0=0.47$, a pronounced feedforward hierarchy) and rose to near the recurrent limit by gamma ($F_0=0.87$, i.e.\ only marginally hierarchical; Fig~\ref{fig:ecog}B). Against a direction-shuffled null---preserving each pair's undirected interaction strength but scrambling its direction---theta was overwhelmingly significant ($z=-15.5$, $p<0.001$, Bonferroni-corrected). The higher bands also cleared this tight null ($z$ from $-8.6$ to $-5.5$) but only weakly in absolute terms, and this residual high-band gradient may partly reflect the completion step; the dominant and interpretable effect is a theta-band feedforward hierarchy. The theta hierarchy had an interpretable polarity (Fig~\ref{fig:ecog}C): ventromedial and paralimbic cortex---olfactory cortex, temporal pole, insula, and medial and orbital prefrontal cortex---occupied the source end, while lateral prefrontal and temporal cortex and early visual areas (calcarine, lingual) formed the sinks. The hippocampal formation, the apex of the rodent theta hierarchy, is subcortical and excluded here, but the cortical source end is precisely the paralimbic system with which it is most tightly coupled. Thus the theta-band feedforward hierarchy that SPIDER uncovered in the rodent brain re-emerges in human resting cortex, recovered from single-patient recordings that were never acquired together---a replication across species and recording modality that no individual dataset could yield.

\begin{figure}[h!]
  \centering
  \includegraphics[width=\linewidth]{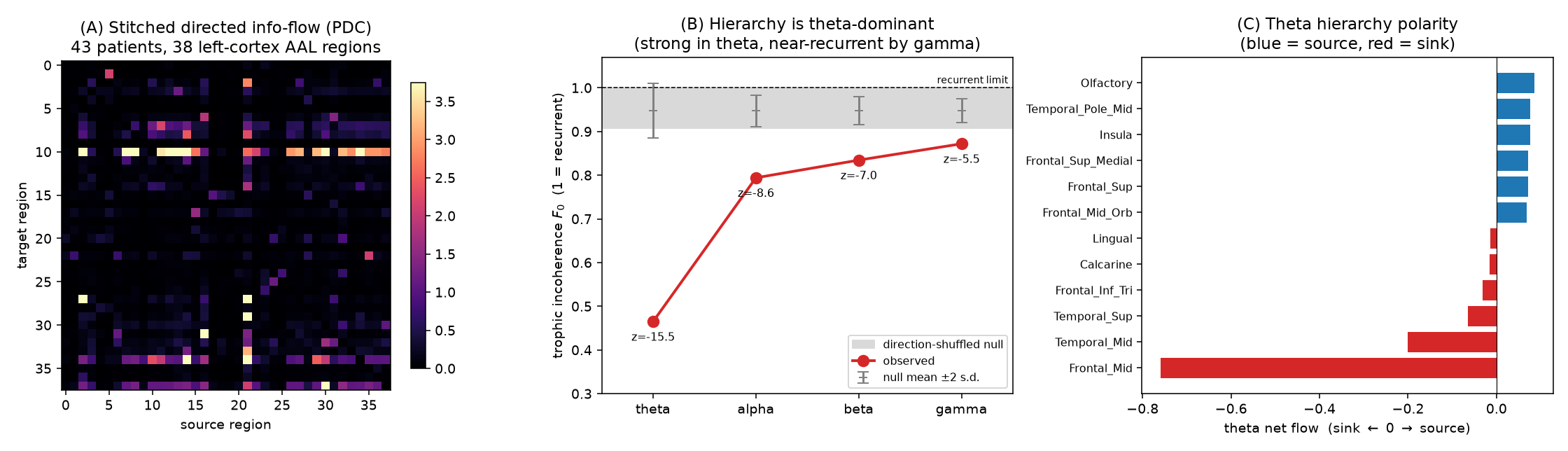}
  \caption{{\bf SPIDER recovers a theta-band feedforward hierarchy in resting human intracranial EEG.}
  Resting intracranial EEG from 43 patients (OpenNeuro \texttt{ds003688}), stitched over 38 left-hemisphere cortical AAL regions that were never recorded simultaneously.
  (A) Stitched directed information-flow (PDC) matrix; color encodes integrated directed flow from source (column) to target (row), with the diagonal zeroed.
  (B) Trophic incoherence $F_0$ per frequency band (red) against a direction-shuffled null (grey, mean $\pm$2\,s.d.); $F_0=1$ is fully recurrent and lower $F_0$ is more feedforward. Theta is strongly hierarchical ($F_0=0.47$, $z=-15.5$); higher bands rise toward the recurrent limit ($F_0=0.87$ by gamma) and are only weakly hierarchical despite clearing the tight null (per-band $z$ annotated).
  (C) Polarity of the theta hierarchy: per-region net directed flow (trophic level); regions are AAL atlas labels~\cite{tzourio2002automated} (left hemisphere). Net sources (blue) are ventromedial and paralimbic cortex---olfactory cortex, middle temporal pole, insula, medial and dorsal superior frontal gyrus, and the orbital part of the middle frontal gyrus. Net sinks (red) are lateral neocortex and early visual cortex---middle frontal gyrus (dorsolateral prefrontal), middle and superior temporal gyri, inferior frontal gyrus (pars triangularis), calcarine cortex (primary visual cortex, V1), and lingual gyrus.}
  \label{fig:ecog}
\end{figure}


\section{Discussion}

We have introduced SPIDER, a non-parametric framework for recovering frequency-domain directed information flow from asynchronous and partially observed neural recordings. The central practical advance for neuroscience is that the method removes the two requirements that have confined directed effective-connectivity analysis to single, fully simultaneous recordings: it needs neither simultaneous observation of all regions nor a shared temporal reference. Because the estimate is frequency-resolved, it captures directed interactions at distinct timescales rather than reducing each connection to a single summary value. As a result, it can reveal band-specific patterns of influence, such as theta- versus gamma-range interactions, that are inaccessible to time-domain or scalar connectivity measures. Under the completeness condition, the stitched spectral estimator and the resulting non-parametric PDC are uniformly consistent for all variable pairs (Theorem~\ref{thm:complete_short}). When completeness fails, consistent recovery remains possible if a structured matrix completion step succeeds (Theorem~\ref{thm:incomplete}); we identify spikiness of the spectrum, sufficient pairwise coverage, and near-low-rank structure as one concrete set of conditions under which nuclear-norm minimization meets that requirement (Proposition~\ref{prop:nnm-sufficient}). The simulation results confirm these theoretical predictions: NNM completion outperforms a naive zero-fill baseline at low to moderate missingness, GLASSO regularization reduces error when the underlying connectivity is sparse, and directed structure is recovered accurately for observed pairs and partially for imputed ones, degrading as pairwise coverage falls. The application to the Allen Brain Observatory data demonstrates that the method scales to real neural recordings of moderate dimension and recovers the dominant directional structure even with substantial missingness. Applied to the IBL brain-wide Neuropixels dataset, the area-level variant recovers directed information flow across 50 brain regions by stitching data from 43 sessions collected in 12 laboratories, none of which were recorded simultaneously. The strongest inferred interactions show substantial correspondence with anatomical projections from the Allen Institute for Brain Science, achieving ROC--AUC values approaching $0.9$ for observed region pairs. Applied to resting human intracranial EEG, the same pipeline reconstructed a directed whole-cortex network by pooling 43 patients with non-overlapping electrode coverage, and recovered the same theta-band feedforward hierarchy---significant against a direction-shuffled null---demonstrating the finding across species and recording modality. Across these applications the method operated at the tens-to-hundreds-channel scale of current recordings; the asymptotic regime in which the number of variables greatly exceeds the per-block sample size is characterized in simulation.

\subsection{A frequency-specific hierarchy of spontaneous flow}

Beyond validating the estimator, the IBL analysis yields a concrete neuroscientific result: integrated over frequency, the brain-wide directed network is recurrent, but in the theta band it reorganizes into a statistically significant feedforward hierarchy with the hippocampal formation at its source and thalamo-frontal and midbrain structures as sinks, while higher frequencies remain dominated by reciprocal flow (Fig~\ref{fig:hierarchy}). This band-specific directionality---feedforward in theta, recurrent in gamma---is precisely the kind of structure that scalar or time-domain connectivity measures cannot express, and it could only be read out by pooling sessions that were never recorded together. It is consistent with the role of theta in organizing hippocampal--cortical communication~\cite{buzsaki2002theta,colgin2013mechanisms}, the hippocampal formation being both a principal generator of theta and, through the subiculum, its main output stage. The same theta-band hierarchy re-emerges in resting human intracranial EEG (Fig~\ref{fig:ecog}), recovered by pooling patients with non-overlapping coverage, with its cortical source in ventromedial and paralimbic cortex---indicating that the organization is not specific to the rodent or to spike-based recording.

A more tentative observation concerns the apex of the theta hierarchy, the periaqueductal gray (PAG), which received far more directed flow than it emitted. Within our data this sink position rests on limited evidence---the PAG was directly co-observed with only six other areas---and its measured theta input was dominated by the hippocampal formation (notably CA1 and dentate gyrus), consistent with the PAG being entrained as a recipient within the hippocampal theta-coordinated network rather than with any specific hippocampo-PAG projection. This position is at least compatible with the PAG's established role as a convergence hub and obligatory final premotor gate for diverse innate behaviors---from defensive responses~\cite{bandler1994columnar,tovote2016midbrain} to vocalization, where a dedicated PAG population is necessary and sufficient to gate ultrasonic calls in mice~\cite{jurgens2009vocalization,tschida2019vocalization}---a role under which a net-recipient signature is unsurprising. We deliberately avoid a stronger reading: the descending forebrain pathways classically associated with the PAG, from medial prefrontal, cingulate and amygdalar regions, were not co-observed with it here and so could not be assessed, and the PAG's own brainstem premotor targets lie outside the recorded set; the node-level interpretation therefore awaits recordings that co-sample the PAG with its forebrain afferents. We frame the hierarchy as a whole as hypothesis-generating: it rests on passive, spontaneous epochs and on directly co-observed area pairs, it is an effective-connectivity rather than an anatomical claim, and it illustrates that frequency-resolved effective connectivity at brain-wide scale can surface testable hypotheses about the directional architecture of spontaneous activity.

\subsection{Pooling across sessions and animals}

The independence of stitching from temporal alignment is what makes the framework applicable to data pooled across sessions or animals: blocks need not share a clock, and only the marginal spectral statistics within each block enter the estimator. The key remaining requirement when pooling across heterogeneous sources is that blocks be exchangeable samples from approximately the same generative process. In practice this is the standard assumption underlying any pooled analysis of population recordings; it can be checked by comparing within-block spectra across blocks before stitching, and violations can be detected as anomalous local PSD estimates that disagree on overlapping channels.

\subsection{Non-stationarity}

Although the theory assumes wide-sense stationarity, the method extends straightforwardly to mildly non-stationary settings via standard segmentation or windowed spectral estimation, under the assumption that signal statistics vary slowly relative to the window length. Note that this requires a time axis only \emph{within} each block, so that its windows are mutually comparable; it does not reintroduce the shared temporal reference \emph{across} blocks that the method otherwise avoids. This is the same approach taken by virtually all spectral analyses of neural data and does not require modification of the stitching, completion, or regularization steps. When a recording spans changes in brain state, however, within-window non-stationarity biases the spectral estimate; restricting to passive or spontaneous epochs limits but does not eliminate this, and comparing estimates across sub-windows provides a simple diagnostic.

\subsection{Limitations}

The theoretical guarantees presented here are sufficient conditions, and the constants appearing in the scaling requirements are typically unknown and may depend on the data. Practitioners should treat these conditions as order-of-magnitude guides rather than sharp thresholds and, where possible, use diagnostics---such as the condition number of the stitched matrix, cross-validation of the GLASSO penalty, and the completion residual on held-out co-observed entries---to assess whether the method is operating in a reliable regime. With this caveat, four failure modes are worth flagging explicitly.

\textbf{Long-memory processes.} The local PSD consistency on which the entire pipeline rests requires exponential decay of the dependence-adjusted norm of each component series~\cite{zhang2025spectral}. Processes with polynomially decaying correlations, such as long-memory or fractionally integrated processes, fall outside this condition, and spectral estimation may be unreliable in those regimes. The downstream PDC consistency results no longer apply.

\textbf{Insufficient pairwise coverage.} NNM completion requires the coverage condition $|\mathcal{O}| \gtrsim r K \log K \cdot \log T$ of Proposition~\ref{prop:nnm-sufficient}. If the recording design yields fewer co-observed pairs, equivalently, if the per-pair observation probability satisfies $p_{\mathrm{obs}} \ll r \log K \cdot \log T / K$, completion is not guaranteed to recover the missing entries, and PDC estimates for unobserved pairs will be unreliable. The simulation results show that performance on unobserved pairs degrades sharply as missingness exceeds approximately $70\%$, consistent with this prediction.

\textbf{High effective rank.} If $S(\omega)$ does not admit a good low-rank approximation, for example, when activity is driven by many independent sources of comparable strength rather than a small number of dominant ones, the approximation error in the completion bound does not vanish and the near-low-rank assumption fails. In this regime, completion provides little benefit over zero-filling.

\textbf{High correlation and ill-conditioning.} When observed signals are highly correlated, the spectral matrix is near-singular and any method that requires its inversion or factorization, including ours, becomes numerically unstable. This is a fundamental difficulty rather than a limitation specific to stitching. Strong external sources driving the dynamics, or redundant channels recording overlapping signals, are typical causes. Pre-processing such data by clustering channels or extracting latent components prior to causal analysis is a sensible, practical mitigation.

\subsection{Future directions}

Several extensions are natural. First, alternative completion algorithms with different structural priors, such as nonconvex factorizations, weighted nuclear-norm variants, and explicit factor-model approaches, could be substituted for the nuclear-norm program in Step~2 of Algorithm~\ref{alg:stitched_pdc}, with the conditional structure of Theorem~\ref{thm:incomplete} ensuring that any such substitution that achieves uniform completion consistency yields PDC consistency. Second, the framework could be combined with hypothesis-testing procedures to produce calibrated confidence statements on individual directed edges, building on recent work in high-dimensional spectral inference. Third, because the pipeline decomposes into independent computations at each frequency, it is readily amenable to GPU acceleration, which could substantially expand the dimensionality of networks that can be analyzed in practice.

\section{Conclusion}
By eliminating the need for both simultaneous observation of all variables and a shared temporal reference, SPIDER enables systematic whole-brain analysis of multi-session and multi-animal datasets that could not previously be combined for directed-flow inference. It turns the fragmented, asynchronous recordings that brain-wide initiatives already produce into a single, frequency-resolved map of directed effective connectivity, and does so non-parametrically for continuous signals, spike trains, and their mixtures alike. We expect this to broaden the range of experimental designs from which whole-brain information flow can be studied.

\section*{Funding}
D.Y.T. was supported by CNPq grants 421955/2023-6 and 445096/2024-1, FAPESP grant 2013/07699-0, and Serrapilheira grant R-2401-47364. Y.S.Z. thanks CAPES-PRINT UFRN 88887.979326/2024-00, for supporting a one-month visit to the Brain Institute, UFRN, and National Natural Science Foundation of China (grant 2025ZD0215700).

\section*{Data and code availability}
The two-photon calcium-imaging data are publicly available through the Allen Brain Observatory~\cite{de2020large} at \url{https://observatory.brain-map.org/visualcoding/search/cell_list?experiment_container_id=511510989&sort_field=p_sg&sort_dir=asc}.
The Neuropixels electrophysiology data through the International Brain Laboratory brain-wide map~\cite{international2025brain}. The human intracranial EEG data are publicly available on OpenNeuro (\texttt{ds003688})~\cite{berezutskaya2022open} at \url{https://openneuro.org/datasets/ds003688}. Code implementing the SPIDER pipeline, together with the processed spectral matrices and area annotations required to reproduce the figures, can be found at \url{https://github.com/yisiszhang/stitched_pdc.git}.

\section*{Declaration of generative AI use}
During the preparation of this work, the authors used ChatGPT and Claude to help with coding and text editing. The authors reviewed and edited the output as needed and take full responsibility for the content of the published article.


\appendix

\section{Mathematical details of the SPIDER pipeline}
\label{S_mathdetails}

This appendix collects the equations underlying the narrative in Materials and Methods. Appendix~\ref{S1_Appendix} and the following sections give the formal process classes, theorems, and proofs.

\subsection{Partial directed coherence and its non-parametric form}

Consider the $K$-dimensional time series $\{X(t)\}_{t \in \mathbb{Z}}$ with $X(t)=[X_1(t),\ldots,X_K(t)]^T$ and $\mathbb{E}[X(t)] = 0$. The process admits a VAR representation of order $p \in \mathbb{N} \cup \{+\infty\}$ if
\begin{equation} \label{eq:AR}
X(t) = \sum_{s=1}^p A(s)\,X(t-s) + \varepsilon(t),
\end{equation}
where the $A(s)$ are $K \times K$ matrices and $\{\varepsilon(t)\}$ is an innovation process with $\mathbb{E}[\varepsilon(t)\varepsilon(s)^T]=0$ for $t \neq s$ and $\mathbb{E}[\varepsilon(t)\varepsilon(t)^T] = \Sigma = \Sigma^T$. The autoregressive transfer function is $\bar{A}(\omega) = I - \sum_{s=1}^p A(s)\, e^{i \omega s}$ for $\omega \in [-\pi,\pi)$. The informational PDC from series $\ell$ to series $k$ (Eq.~\eqref{eq:pdc} of the main text) is $\pi_{k\leftarrow \ell}(\omega) = \bar{A}_{k\ell}(\omega)\Sigma^{-1/2}_{kk} / \sqrt{\bar{A}_{\cdot \ell}(\omega)^H \Sigma^{-1} \bar{A}_{\cdot \ell}(\omega)}$, and the integrated information flow $M_{k\leftarrow\ell}$ is given by Eq.~\eqref{eq:infoflow}.

For the non-parametric formulation we assume a $K \times K$ PSD matrix $S(\omega)$ that is bounded uniformly from below and above: for some constant $c>0$ and almost every $\omega$,
\begin{equation} \label{eq:bdd-spectrum}
c^{-1} I \;\leq\; S(\omega) \;\leq\; c\, I,
\end{equation}
where $A \leq B$ means $B-A$ is positive semi-definite. Together with pure non-determinism this yields a canonical spectral factorization of the inverse spectrum: for almost all $\omega$,
\begin{equation} \label{eq:can-fact}
S(\omega)^{-1} \;=\; P(\omega) \;=\; \bar{A}(\omega)^H\, \Sigma^{-1}\, \bar{A}(\omega),
\end{equation}
equivalently $S(\omega) = H(\omega)\,\Sigma\, H(\omega)^H$ with $H(\omega) = \bar{A}(\omega)^{-1}$ minimum-phase (the form produced by Wilson's algorithm~\cite{wilson1978convergence}). The PDC of~\eqref{eq:pdc} can then be written using only the inverse PSD and its factorization,
\begin{equation} \label{eq:nonparamPDC}
\pi_{k\leftarrow \ell}(\omega) \;=\; \frac{\bar{A}_{k\ell}(\omega)\,\Sigma^{-1/2}_{kk}}{P_{\ell\ell}^{1/2}(\omega)},
\end{equation}
making explicit that PDC requires no finite-order autoregressive model.

\subsection{Regularization, matrix completion, and consistency ingredients}

\paragraph{Local consistency and high-dimensional scaling}
The per-block estimators are required to be locally consistent.
\begin{assumption}[Local spectral consistency] \label{ass:local}
For all $u \in [m]$ and $(i,j) \in \Omega_u \times \Omega_u$,
\begin{equation}
\sup_{\omega \in [-\pi,\pi)} \bigl|\hat{S}^{\Omega_u}_{ij}(\omega) - S^{\Omega_u}_{ij}(\omega)\bigr| \;\xrightarrow{p}\; 0
\quad \text{as } T_u \to \infty,
\end{equation}
where $T_u$ is the length of block $u$.
\end{assumption}
This holds under standard mixing and moment conditions~\cite{walden2000unified}; in high dimensions, the smoothed periodogram is consistent across all entries and frequencies with dimension allowed to grow nearly exponentially with sample size,
\begin{equation}\label{eq:zz-scaling}
\log K \;=\; o(T_u^{C}) \quad \text{for some constant } C > 0,
\end{equation}
provided each component process has an exponentially decaying dependence-adjusted norm~\cite{zhang2025spectral}, with recommended bandwidth $B=\mathcal{O}(T_u^{1/2}/\log K)$.

\paragraph{Complex Hermitian GLASSO}
Because the stitched estimator $\hat{S}(\omega)$ need not be positive definite, we map the complex Hermitian matrix to its equivalent real symmetric $2K \times 2K$ form
\begin{equation}\label{eq:augmented}
C(\omega) \;=\;
\begin{pmatrix}
\mathrm{Re}\,\hat{S}(\omega) & -\mathrm{Im}\,\hat{S}(\omega) \\
\mathrm{Im}\,\hat{S}(\omega) & \phantom{-}\mathrm{Re}\,\hat{S}(\omega)
\end{pmatrix} \;\in\; \mathbb{R}^{2K \times 2K},
\end{equation}
and solve the real-symmetric GLASSO
\begin{equation}\label{eq:glasso}
\hat{\Theta}(\omega) \;=\; \arg\min_{\Theta \succ 0}\; \bigl\langle \Theta,\, C(\omega)\bigr\rangle \;-\; \log\det \Theta \;+\; \lambda\, \|\Theta\|_{1,\mathrm{off}},
\end{equation}
with $\langle A, B \rangle = \mathrm{tr}(A^T B)$ and $\|\Theta\|_{1,\mathrm{off}} = \sum_{i \neq j} |\Theta_{ij}|$. The regularized estimator is obtained by recombining the appropriate $K \times K$ blocks of $\hat{\Theta}(\omega)$ and inverting,
\begin{equation}\label{eq:tilde-S}
\tilde{S}(\omega) \;=\; \bigl([\hat{\Theta}(\omega)]_{1:K,\,1:K} + i\,[\hat{\Theta}(\omega)]_{K+1:2K,\,1:K}\bigr)^{-1},
\end{equation}
which is Hermitian positive definite with a stable inverse.
\begin{definition}[SPIDER estimator under completeness]\label{def:stitchedPDC}
Assume Condition~\ref{cond:complete} holds, and let $\tilde{S}(\omega)$ be the regularized stitched power spectrum. The \emph{SPIDER estimator} is obtained by applying the canonical spectral factorization to $\tilde{S}(\omega)$ and plugging the resulting factors into~\eqref{eq:nonparamPDC}.
\end{definition}

\paragraph{Near-low-rank prior and matrix completion}
The completion step assumes that population activity is driven by $r \ll K$ latent sources $Z(t)$ through a linear mixing model
\begin{equation}\label{eq:mixing}
X(t) = A\, Z(t) + \varepsilon(t),
\qquad A \in \mathbb{C}^{K \times r},
\end{equation}
with independent low-power sensor noise $\varepsilon(t)$, so that
\begin{equation}\label{eq:lowrank-plus-noise}
S(\omega) = A\, S_Z(\omega)\, A^H + S_\varepsilon(\omega),
\end{equation}
which is full rank but approximately low rank. In the incomplete case the stitched matrix is only partially specified.
\begin{definition}[Stitched PSD under incomplete observation]\label{def:stitched_incomplete}
Let $\{\Omega_u\}_{u=1}^m$ satisfy $\bigcup_{u=1}^m \Omega_u = [K]$ but not necessarily Condition~\ref{cond:complete}, let $\mathcal{O} = \{(i,j): \exists\, u,\ i,j \in \Omega_u\}$ be the set of observed pairs, and let $N_{ij} = \#\{u : i,j \in \Omega_u\}$. The stitched PSD is
\[
\hat{S}_{ij}(\omega) \;=\;
\begin{cases}
\displaystyle \frac{1}{N_{ij}} \sum_{u : i,j \in \Omega_u} \hat{S}^{\Omega_u}_{ij}(\omega), & (i,j) \in \mathcal{O},\\[2ex]
\text{unobserved}, & (i,j) \notin \mathcal{O}.
\end{cases}
\]
\end{definition}
Missing entries are recovered by nuclear-norm minimization over Hermitian positive definite matrices (Algorithm~\ref{alg:stitched_pdc}, Step~2). Modeling the observation pattern as random, with each pair observed independently with probability $p_{\mathrm{obs}}$, the coverage condition of Proposition~\ref{prop:nnm-sufficient} becomes
\begin{equation}\label{eq:coverage-prob}
p_{\mathrm{obs}} \;\gtrsim\; \frac{r\, \log K \cdot \log T}{K},
\qquad T = \min_u T_u,
\end{equation}
which is mild whenever $r \ll K$.

\subsection{Area-level cross-spectra via multitaper estimation and spectral PCA}

Let neuron $i$ in area $a$, session $s$, fire at times
$\mathcal{T}_i = \{t_{i,1},\ldots,t_{i,n_i}\}\subset[0,T_s]$.
Compute $L = 2NW-1$ discrete prolate spheroidal sequences (DPSS)
$\{w_l(t)\}_{l=1}^{L}$ of length $T_s$ with time-bandwidth product $NW$, with taper Fourier transforms
\begin{equation}
    H_l(\omega)
    = \int_0^{T_s} w_l(t)e^{-i\omega t}\,\mathrm{d}t.
    \label{eq:dpss_ft}
\end{equation}
The mean-rate-corrected tapered transform of neuron $i$ under taper $l$ is
\begin{equation}
    \hat{J}^{(l)}_{i}(\omega)
    = \sum_{j=1}^{n_i}
      w_l(t_{i,j})e^{-i\omega t_{i,j}}
      - H_l(\omega)\,\bar{\lambda}_i,
    \label{eq:tapered_ft}
\end{equation}
where $\bar{\lambda}_i = n_i/T_s$ is the mean firing rate, and the neuron-level cross-spectral matrix
$\hat{\mathbf S}^{(s)}_{ab}(\omega)\in\mathbb C^{N_{as}\times N_{bs}}$ has entries
\begin{equation}
    \bigl[\hat{\mathbf S}^{(s)}_{ab}(\omega)\bigr]_{ij}
    =
    \frac{1}{L\,T_s}
    \sum_{l=1}^{L}
    \overline{\hat J^{(l)}_{i}(\omega)}
    \hat J^{(l)}_{j}(\omega),
    \qquad i\in a,\; j\in b.
    \label{eq:neuron_csd}
\end{equation}

\begin{definition}[Area-level canonical cross-spectrum]
For areas $a,b\in\Omega_s$, the area-level canonical cross-spectrum is the
bilinear projection of the neuron-level PSD matrix onto the leading
spectral eigenvectors:
\begin{equation}
    \tilde{S}^{(s)}_{ab}(\omega)
    = \bigl[\mathbf{v}^{(s)}_a(\omega)\bigr]\Herm
      \hat{\mathbf{S}}^{(s)}_{ab}(\omega)\;
      \mathbf{v}^{(s)}_b(\omega)
    \;\in\;\C,
    \label{eq:projection}
\end{equation}
where $\mathbf{v}^{(s)}_a(\omega)\in\C^{N_{as}}$ (spectral PC1) is the leading eigenvector of the within-area spectral decomposition
\begin{equation}
    \hat{\mathbf{S}}^{(s)}_{aa}(\omega)
    = \mathbf{V}^{(s)}_a(\omega)\,
      \boldsymbol{\Lambda}^{(s)}_a(\omega)\,
      \bigl[\mathbf{V}^{(s)}_a(\omega)\bigr]\Herm .
    \label{eq:eig}
\end{equation}
\end{definition}
The phase ambiguity in $\mathbf{v}^{(s)}_a(\omega)$ is fixed so that the element of maximum absolute magnitude is real and positive,
\begin{equation}
    \mathbf{v}^{(s)}_a(\omega)
    \;\leftarrow\;
    \frac{\overline{[\mathbf{v}^{(s)}_a(\omega)]_{i^*}}}
         {\bigl|[\mathbf{v}^{(s)}_a(\omega)]_{i^*}\bigr|}\;
    \mathbf{v}^{(s)}_a(\omega),
    \qquad
    i^* = \arg\max_i \bigl|[\mathbf{v}^{(s)}_a(\omega)]_i\bigr| .
    \label{eq:sign}
\end{equation}
An area is retained only if its frequency-averaged PC1 explained variance exceeds a threshold,
\begin{equation}
    \bar{\rho}^{(s)}_a
    =
    \frac{1}{|W|}
    \sum_{\omega\in W}
    \rho^{(s)}_a(\omega)
    \;\geq\;
    \rho_{\min},
    \label{eq:quality}
\end{equation}
with $\rho_{\min}=0.10$ here. Finally, to make local estimates commensurable across sessions, the area-level cross-spectra are normalized to coherence,
\begin{equation}
    \tilde{C}^{(s)}_{ab}(\omega)
    = \frac{\tilde{S}^{(s)}_{ab}(\omega)}
           {\sqrt{\tilde{S}^{(s)}_{aa}(\omega)\;\tilde{S}^{(s)}_{bb}(\omega)}},
    \label{eq:coherence}
\end{equation}
before assembly via~\eqref{eq:agrrespec} and factorization.

\subsection{Algorithm}
\label{algorithm1}

\begin{figure}[H]
\algcaption{Stitched non-parametric partial directed coherence (SPIDER). When Condition~\ref{cond:complete} holds, skip Step~2.}\label{alg:stitched_pdc}
\begin{mdframed}
\begin{algorithmic}[1]
\REQUIRE Asynchronous blocks $\{X_{\Omega_u}\}_{u=1}^{m}$ with $\bigcup_{u=1}^{m} \Omega_u = [K]$; GLASSO penalty $\lambda \geq 0$; completion penalty $\mu > 0$; projection $\mathcal{P}_{\mathcal{O}}$ onto observed matrix entries; nuclear norm $\|\cdot\|_*$; Frobenius norm $\|\cdot\|_F$.
\ENSURE SPIDER whole-network PDC estimate $\hat{\pi}_{k\leftarrow \ell}(\omega)$ for $k, \ell \in [K]$.
\STATE \textbf{Step 1: Spectral aggregation}
\STATE Compute local PSD estimates $\hat{S}^{\Omega_u}(\omega)$ for each block (e.g., via multi-taper spectral estimation).
\STATE Construct the stitched PSD matrix $\hat{S}(\omega)$ as in~\eqref{eq:agrrespec}.
\STATE \textbf{Step 2: Structured matrix completion} \emph{(skip if Condition~\ref{cond:complete} holds)}
\STATE Map the complex stitched estimate $\hat{S}(\omega)$ to its augmented real-symmetric form
\[
C(\omega) =
\begin{pmatrix}
\mathrm{Re}\,\hat{S}(\omega) & -\mathrm{Im}\,\hat{S}(\omega) \\
\mathrm{Im}\,\hat{S}(\omega) & \phantom{-}\mathrm{Re}\,\hat{S}(\omega)
\end{pmatrix}.
\]
\STATE Solve the constrained nuclear-norm program
\[
\tilde{C}(\omega) \;=\; \arg\min_{C \succeq 0}\; \|C\|_* \;+\; \frac{\mu}{2}\, \bigl\|\mathcal{P}_{\mathcal{O}}\bigl(C - C(\omega)\bigr)\bigr\|_F^2.
\]
\STATE Reconstruct $\hat{S}^{c}(\omega) = [\tilde{C}]_{1{:}K,\,1{:}K} + i\,[\tilde{C}]_{K+1{:}2K,\,1{:}K}$.
\STATE \textbf{Step 3: Hermitian GLASSO regularization}
\STATE Map $\hat{S}^{c}(\omega)$ to its augmented real-symmetric form $C^{c}(\omega)$ as in~\eqref{eq:augmented} (with $\hat{S}^{c}=\hat{S}$ when Step~2 is skipped) and solve
\[
\hat{\Theta}(\omega) = \arg\min_{\Theta \succ 0}\; \langle \Theta,\, C^{c}(\omega) \rangle - \log\det \Theta + \lambda\,\|\Theta\|_{1,\mathrm{off}}.
\]
\STATE Recombine the $K \times K$ blocks of $\hat{\Theta}(\omega)$ into a Hermitian inverse spectrum and invert to obtain $\tilde{S}(\omega)$ as in~\eqref{eq:tilde-S}.
\STATE \textbf{Step 4: Canonical spectral factorization}
\STATE Factorize $\tilde{S}(\omega) = \hat{H}(\omega)\,\hat{\Sigma}\,\hat{H}(\omega)^{H}$ using Wilson's algorithm.
\STATE Compute $\hat{\bar{A}}(\omega) = \hat{H}(\omega)^{-1}$.
\STATE \textbf{Step 5: Plug-in PDC estimator}
\STATE \textbf{Return} the SPIDER estimate
\[
\hat{\pi}_{k\leftarrow \ell}(\omega) \;=\; \frac{\hat{\bar{A}}_{k\ell}(\omega)\, \hat{\Sigma}_{kk}^{-1/2}}{\sqrt{\,\hat{\bar{A}}_{\cdot \ell}(\omega)^{H}\, \hat{\Sigma}^{-1}\, \hat{\bar{A}}_{\cdot \ell}(\omega)\,}}.
\]
\end{algorithmic}
\end{mdframed}
\end{figure}

\subsection{Hierarchy of directed flow: trophic incoherence and flow decomposition}\label{app:hierarchy}

To characterize the directed network recovered in a given frequency band we use the integrated information-flow matrix $M$ of~\eqref{eq:infoflow}, writing $W_{ij}=M_{j\leftarrow i}$ for the flow from area $i$ to area $j$ and restricting to directly co-observed pairs. Following MacKay, Johnson and Sansom~\cite{mackay2020directed,johnson2014trophic}, each area is assigned a \emph{trophic level} $h\in\mathbb{R}^{K}$ as the solution of
\[
\bigl(\mathrm{diag}(u+v) - (W+W^{\top})\bigr)\,h \;=\; u - v, \qquad u_i=\textstyle\sum_j W_{ji},\quad v_i=\textstyle\sum_j W_{ij},
\]
with $h$ centered to fix the additive constant, where $u$ and $v$ are the weighted in- and out-strengths. The \emph{trophic incoherence}
\[
F_0 \;=\; \frac{\sum_{ij} W_{ij}\,(h_j-h_i-1)^2}{\sum_{ij} W_{ij}} \;\in\;[0,1]
\]
measures the flow-weighted dispersion of edges around a perfect one-level feedforward step: $F_0=0$ exactly when every edge ascends one level (a perfect hierarchy), and $F_0\to1$ as directions become unstructured. Significance is assessed against a null ensemble that, for each unordered pair, swaps the two directed weights with probability one half, preserving all undirected strengths while destroying directional organization; we report the $z$-score and empirical $p$-value of the observed $F_0$ relative to this ensemble.

To separate kinds of feedback we decompose the directed flow into three nonnegative shares of the total throughput $\sum_{ij}W_{ij}$. The \emph{reciprocal} (mutual) share is $\sum_{i<j} 2\min(W_{ij},W_{ji})$. The remaining net flow $F_{ij}=W_{ij}-W_{ji}$ is split by a discrete Helmholtz--Hodge decomposition into a \emph{gradient} (feedforward) part, the flow-weighted least-squares fit of $F_{ij}$ by the level difference $h_j-h_i$, and a divergence-free \emph{circulating} (loop) part; the fraction of net-flow energy in the circulating part is the cyclicity, and its complement the feedforward gradient share. These three shares are reported per band in Fig~\ref{fig:hierarchy}D.

\section{Formal characterization of admissible processes}
\label{S1_Appendix}
\label{S_admissible}

This section makes precise the class of processes to which the non-parametric
pipeline applies and states the condition under which the canonical spectral
factorization used throughout the main text exists. The spectra below are
written with the normalization $S(\omega)=\frac{1}{2\pi}\sum_{\tau\in\mathbb{Z}}
\Gamma(\tau)\,e^{-i\omega\tau}$, where $\Gamma(\tau)$ is the (matrix) covariance
sequence; any overall positive scalar in this normalization cancels in the PDC
of~\eqref{eq:nonparamPDC} and~\eqref{eq:pdc}, so the choice is immaterial for the
estimand.

\subsection{Canonical spectral factorization}
\label{S_factorization}

\begin{definition}[Admissible spectrum]
\label{def:admissible}
A $K\times K$ matrix-valued function $S:[-\pi,\pi)\to\C^{K\times K}$ is an
\emph{admissible spectrum} if (i) $S(\omega)=S(\omega)\Herm$ for almost every
$\omega$; (ii) $S$ is integrable; and (iii) there is a constant $c>0$ with
$c^{-1} I \preceq S(\omega)\preceq c\,I$ for almost every $\omega$, where
$A\preceq B$ means $B-A$ is positive semi-definite.
\end{definition}

Condition (iii) is exactly the boundedness assumption~\eqref{eq:bdd-spectrum} of
the main text; it implies in particular that $S(\omega)$ is Hermitian positive
definite almost everywhere and that $\log\det S(\omega)$ is bounded, hence
integrable.

\begin{proposition}[Existence and uniqueness of the canonical factorization]
\label{prop:wienermasani}
Let $S$ be an admissible spectrum. Then there exist a matrix-valued function
$H(\omega)$ and a constant Hermitian positive definite matrix $\Sigma$ such that
\begin{equation}
S(\omega)=H(\omega)\,\Sigma\,H(\omega)\Herm \quad\text{for a.e. }\omega,
\label{eq:S_admissible_fact}
\end{equation}
where $H$ is the boundary value of a function analytic and non-singular on the
open unit disc with $H(0)=I$ (the \emph{minimum-phase} or \emph{canonical}
factor). The pair $(H,\Sigma)$ is unique. Equivalently,
$S(\omega)^{-1}=\bar{A}(\omega)\Herm\,\Sigma^{-1}\,\bar{A}(\omega)$ with
$\bar{A}(\omega)=H(\omega)^{-1}$, which is~\eqref{eq:can-fact}.
\end{proposition}

This is the multivariate Wiener--Hopf--Masani factorization theorem; the relevant
sufficient condition is the Szeg\H{o} requirement
$\int_{-\pi}^{\pi}\log\det S(\omega)\,d\omega>-\infty$, which holds automatically
under Definition~\ref{def:admissible}(iii). Wilson's algorithm~\cite{wilson1978convergence}
computes $(H,\Sigma)$ numerically and is the procedure used in
Algorithm~\ref{alg:stitched_pdc}. The non-parametric PDC of~\eqref{eq:nonparamPDC}
is a continuous functional of $(\bar{A},\Sigma)$ and is therefore well-defined for
every admissible spectrum, independently of any finite-order parametric
representation.

\subsection{Linear processes: VAR$(\infty)$, VMA, and VARMA}
\label{S_linear}

\begin{definition}[Vector linear processes]
Let $\{\varepsilon(t)\}_{t\in\mathbb{Z}}$ be a $K$-dimensional white-noise
innovation sequence with $\mathbb{E}[\varepsilon(t)]=0$,
$\mathbb{E}[\varepsilon(t)\varepsilon(t)\T]=\Sigma\succ0$, and
$\mathbb{E}[\varepsilon(t)\varepsilon(s)\T]=0$ for $t\neq s$. A zero-mean
wide-sense stationary process $\{X(t)\}$ is:
\begin{itemize}
\item a \emph{$\mathrm{VMA}(\infty)$ process} if
$X(t)=\sum_{s\geq0}B(s)\,\varepsilon(t-s)$ with $\sum_{s\geq0}\|B(s)\|<\infty$ and
$B(0)=I$;
\item a \emph{$\mathrm{VAR}(\infty)$ process} if
$X(t)=\sum_{s\geq1}A(s)\,X(t-s)+\varepsilon(t)$ with $\sum_{s\geq1}\|A(s)\|<\infty$;
\item a \emph{$\mathrm{VARMA}(p,q)$ process} if
$X(t)-\sum_{s=1}^{p}A(s)X(t-s)=\varepsilon(t)+\sum_{s=1}^{q}B(s)\varepsilon(t-s)$.
\end{itemize}
\end{definition}

Writing $\bar{A}(\omega)=I-\sum_{s\geq1}A(s)e^{-i\omega s}$ and
$\bar{B}(\omega)=I+\sum_{s\geq1}B(s)e^{-i\omega s}$, each of these processes has
spectrum
\begin{equation}
S(\omega)=\frac{1}{2\pi}\,\bar{A}(\omega)^{-1}\bar{B}(\omega)\,\Sigma\,
\bar{B}(\omega)\Herm\bar{A}(\omega)^{-H},
\label{eq:varma_spectrum}
\end{equation}
with $\bar{B}\equiv I$ for the VAR case and $\bar{A}\equiv I$ for the VMA case.

\begin{condition}[Stability and minimum phase]
\label{cond:stability}
$\det\bar{A}(z)\neq0$ and $\det\bar{B}(z)\neq0$ for all $|z|\leq1$, where
$z=e^{-i\omega}$ is extended into the closed unit disc.
\end{condition}

Under Condition~\ref{cond:stability} together with absolute summability, the
spectrum~\eqref{eq:varma_spectrum} is an admissible spectrum in the sense of
Definition~\ref{def:admissible}, and its canonical factor coincides with
$H(\omega)=\bar{A}(\omega)^{-1}\bar{B}(\omega)$. The infinite-order VAR
representation is the natural data-generating model for the high-dimensional
consistency result invoked in the main text (Eq~\eqref{eq:zz-scaling}), and the
exponential decay of $\|A(s)\|$ assumed there is a special case of absolute
summability.

\subsection{Point processes and the Bartlett spectrum}
\label{S_point}

Let $N=(N_1,\dots,N_K)$ be a $K$-variate simple, stationary point process on
$\mathbb{R}$, where $N_k(\mathrm{d}t)$ counts events of type $k$. Write
$\lambda_k=\mathbb{E}[N_k(\mathrm{d}t)]/\mathrm{d}t$ for the (constant) mean
intensity and let
\begin{equation}
\gamma_{ij}(\tau)\,\mathrm{d}\tau
=\mathrm{Cov}\!\bigl(N_i(\mathrm{d}t),\,N_j(\mathrm{d}t+\tau)\bigr)
-\delta_{ij}\,\lambda_i\,\delta(\tau)\,\mathrm{d}\tau
\end{equation}
denote the reduced second-order covariance density, where the Dirac term removes
the atom that each process contributes to its own covariance at $\tau=0$.

\begin{definition}[Bartlett spectrum]
The \emph{Bartlett spectrum} of $N$ is the $K\times K$ Hermitian matrix measure
with density
\begin{equation}
S_{ij}(\omega)=\frac{1}{2\pi}\Bigl(\delta_{ij}\,\lambda_i
+\int_{-\infty}^{\infty}\gamma_{ij}(\tau)\,e^{-i\omega\tau}\,\mathrm{d}\tau\Bigr).
\label{eq:bartlett}
\end{equation}
\end{definition}

\begin{definition}[Linear multivariate Hawkes process]
$N$ is a \emph{linear multivariate Hawkes process} with baseline
$\mu\in\mathbb{R}_{>0}^{K}$ and excitation kernels
$G:\mathbb{R}_{\geq0}\to\mathbb{R}^{K\times K}$ if its conditional intensity is
\begin{equation}
\lambda_k(t\mid\mathcal{H}_t)=\mu_k+\sum_{j=1}^{K}\int_{(-\infty,t)}G_{kj}(t-s)\,N_j(\mathrm{d}s),
\end{equation}
where $\mathcal{H}_t$ is the history up to time $t$. The process is stable
(stationary with finite intensity) if the spectral radius of
$\int_0^\infty G(u)\,\mathrm{d}u$ is strictly less than one, in which case the
stationary rates are $\Lambda=(I-\int_0^\infty G(u)\,\mathrm{d}u)^{-1}\mu$.
\end{definition}

\begin{proposition}[Bartlett spectrum of the Hawkes process]
\label{prop:hawkes}
For a stable linear multivariate Hawkes process, with
$\hat{G}(\omega)=\int_0^\infty G(u)\,e^{-i\omega u}\,\mathrm{d}u$,
\begin{equation}
S(\omega)=\frac{1}{2\pi}\bigl(I-\hat{G}(\omega)\bigr)^{-1}\diag(\Lambda)
\bigl(I-\hat{G}(\omega)\bigr)^{-H},
\label{eq:hawkes_spectrum}
\end{equation}
which is of the canonical form~\eqref{eq:S_admissible_fact} with
$H(\omega)=(I-\hat{G}(\omega))^{-1}$ and $\Sigma=\frac{1}{2\pi}\diag(\Lambda)$.
\end{proposition}

Equation~\eqref{eq:hawkes_spectrum} is the Bartlett power spectral measure of the
Hawkes process~\cite{bremaud2002power}. Stability guarantees that
$I-\hat{G}(\omega)$ is non-singular for all $\omega$ and analytic with analytic
inverse on the unit disc, so $S$ is admissible and its canonical factor is
$H(\omega)=(I-\hat{G}(\omega))^{-1}$. Consequently $\bar{A}(\omega)=I-\hat{G}(\omega)$
and the PDC of a Hawkes process is the normalization of $I-\hat{G}(\omega)$
defined in~\eqref{eq:nonparamPDC}.

\subsection{Mixed point-process and continuous-valued processes}
\label{S_mixed}

Let $Z=(Y,N)$ consist of a $K_c$-dimensional continuous-valued wide-sense
stationary process $Y$ and a $K_p$-variate point process $N$, jointly
second-order stationary, with $K=K_c+K_p$. The joint second-order structure is
described by the block spectral matrix
\begin{equation}
S(\omega)=
\begin{pmatrix}
S_{YY}(\omega) & S_{YN}(\omega)\\[2pt]
S_{NY}(\omega) & S_{NN}(\omega)
\end{pmatrix},
\end{equation}
where $S_{YY}$ is the ordinary spectral density of $Y$, $S_{NN}$ is the Bartlett
spectrum~\eqref{eq:bartlett} of $N$, and the cross-blocks are the Fourier
transforms of the cross-covariance densities
$\mathrm{Cov}(Y_a(t),N_b(\mathrm{d}t+\tau))/\mathrm{d}\tau$ between the continuous
and event channels~\cite{jarvis2001sampling}. A natural generative model couples
the two through a point process whose stochastic intensity
$\lambda_k(t)=\mu_k+\sum_j\int G_{kj}(t-s)\,N_j(\mathrm{d}s)+\sum_a\int
F_{ka}(t-s)\,Y_a(s)\,\mathrm{d}s$ has an autoregressive form driven by both past
events and the continuous signal. Whenever the resulting joint spectrum $S$ is
admissible in the sense of Definition~\ref{def:admissible}, the canonical
factorization of Proposition~\ref{prop:wienermasani} applies to the full block
matrix and the PDC of~\eqref{eq:nonparamPDC} is well-defined across all
continuous and event channels. Estimation requires only a consistent estimate of
each block: the continuous and Bartlett blocks and their cross-spectra are all
obtained from the multitaper estimator, with the point channels entered as
binary event sequences whose tapering implicitly smooths the underlying
intensity.

\section{Consistency under complete stitching (Theorem~\ref{thm:complete_short})}

\begin{theorem}[Consistency of SPIDER under complete stitching]
\label{thm:complete_short}
Let $\mathbf{X}$ be a $K$-dimensional wide-sense stationary, purely non-deterministic process whose spectral density $S(\omega)$ is bounded uniformly from above and below by positive constants. Assume Condition~\ref{cond:complete} (completeness) and Assumption~\ref{ass:local} (local spectral consistency). Let $\hat{S}(\omega)$ be the stitched estimator of Definition~\ref{def:stitched}, $\tilde{S}(\omega)$ its GLASSO-regularized version (Algorithm~\ref{alg:stitched_pdc}) with regularization parameter $\lambda_T \to 0$ as $\min_u T_u \to \infty$, and $\hat{\pi}_{k\leftarrow \ell}(\omega)$ the SPIDER estimator (Definition~\ref{def:stitchedPDC}). Then for all $k, \ell \in [K]$,
\[
\sup_{\omega \in [-\pi,\pi)} \bigl|\hat{\pi}_{k\leftarrow \ell}(\omega) - \pi_{k\leftarrow \ell}(\omega)\bigr| \;\xrightarrow{p}\; 0.
\]
\end{theorem}

\noindent\textbf{Remark} (Verifiable scaling for high-dimensional regimes).
Assumption~\ref{ass:local} is the abstract requirement of Theorem~\ref{thm:complete_short}, but for practitioners it is useful to note one concrete sufficient condition. Under exponential decay of the dependence-adjusted norm of each component series, the result of~\cite{zhang2025spectral} guarantees uniform local spectral consistency provided~\eqref{eq:zz-scaling} holds. The recommended bandwidth is $B = \mathcal{O}(T_u^{1/2}/\log K)$, which optimally balances the polynomial dependence on $T_u$ and $\log K$. In particular, $K$ may grow nearly exponentially in the per-block sample size, so Theorem~\ref{thm:complete_short} applies in genuinely high-dimensional regimes where $K \gg T_u$ would defeat classical spectral theory.

\begin{proof}
Under Condition~\ref{cond:complete}, every pair $(i,j) \in [K] \times [K]$ is contained in at least one observation block $\Omega_u$. By Assumption~\ref{ass:local}, for each block $\Omega_u$ the local spectral estimator $\hat{S}^{\Omega_u}(\omega)$ converges uniformly in probability to the corresponding submatrix $S^{\Omega_u}(\omega)$. Since each entry $\hat{S}_{ij}(\omega)$ of the stitched estimator is obtained by averaging a finite number of uniformly consistent estimators, it follows that
\[
\sup_{\omega \in [-\pi,\pi)} \|\hat{S}(\omega) - S(\omega)\| \;\xrightarrow{p}\; 0.
\]

Because $S(\omega)$ is uniformly bounded from above and below by positive constants, it is uniformly positive definite. Uniform convergence of $\hat{S}(\omega)$ implies that $\hat{S}(\omega)$ is also uniformly positive definite with probability tending to one. The GLASSO operator with regularization parameter $\lambda_T \to 0$ defines a continuous family of mappings on the space of uniformly positive definite Hermitian matrices, and the limiting mapping at $\lambda_T = 0$ reduces to the identity. Hence
\[
\sup_{\omega \in [-\pi,\pi)} \|\tilde{S}(\omega) - S(\omega)\| \;\xrightarrow{p}\; 0.
\]

Uniform positive definiteness of $S(\omega)$ implies the existence of a canonical minimum-phase spectral factorization $S(\omega) = H(\omega)\,\Sigma\, H(\omega)^H$, with $H(\omega)$ invertible and continuous in $\omega$ and $\bar{A}(\omega) = H(\omega)^{-1}$. By continuity of the spectral factorization mapping, $\tilde{S}(\omega) \to S(\omega)$ uniformly implies that the corresponding estimated factors satisfy
\[
\sup_{\omega} \|\hat{H}(\omega) - H(\omega)\| \;\xrightarrow{p}\; 0,
\qquad
\hat{\Sigma} \;\xrightarrow{p}\; \Sigma.
\]

The non-parametric partial directed coherence $\pi_{k\leftarrow \ell}(\omega)$ is a continuous functional of the inverse spectral factor $\bar{A}(\omega) = H(\omega)^{-1}$ and of $\Sigma$, with normalization terms bounded away from zero by uniform positive definiteness. Continuity of matrix inversion and of the PDC mapping therefore yields
\[
\sup_{\omega \in [-\pi,\pi)} \bigl|\hat{\pi}_{k\leftarrow \ell}(\omega) - \pi_{k\leftarrow \ell}(\omega)\bigr| \;\xrightarrow{p}\; 0
\quad \text{for all } k, \ell \in [K].
\]
\end{proof}

\section{Conditional consistency under incomplete observation (Theorem~\ref{thm:incomplete})}

\begin{theorem}[Conditional consistency of SPIDER under incomplete observation]
\label{thm:incomplete}
Let $\mathbf{X}$ be a $K$-dimensional wide-sense stationary, purely non-deterministic process with spectral density $S(\omega)$ bounded uniformly from below and above by positive constants. Let $\{\Omega_u\}_{u=1}^m$ be observation blocks not necessarily satisfying Condition~\ref{cond:complete}, and let $\mathcal{O}$ denote the set of observed index pairs. Assume Assumption~\ref{ass:local} (local spectral consistency). Let $\hat{S}(\omega)$ be the partially specified stitched estimator of Definition~\ref{def:stitched_incomplete} and let $\hat{S}^{c}(\omega)$ be a completed estimator obtained from $\hat{S}(\omega)$ via a structured matrix completion procedure satisfying
\begin{equation}\label{eq:completion-consistency}
\sup_{\omega \in [-\pi,\pi)} \bigl\|\hat{S}^{c}(\omega) - S(\omega)\bigr\| \;\xrightarrow{p}\; 0.
\end{equation}
Then the SPIDER estimator $\hat{\pi}_{k\leftarrow \ell}(\omega)$ obtained via frequency-wise GLASSO regularization with $\lambda_T \to 0$ satisfies, for all $k, \ell \in [K]$,
\[
\sup_{\omega \in [-\pi,\pi)} \bigl|\hat{\pi}_{k\leftarrow \ell}(\omega) - \pi_{k\leftarrow \ell}(\omega)\bigr| \;\xrightarrow{p}\; 0.
\]
\end{theorem}

\begin{proof}
By assumption, the completed spectral estimator satisfies
\[
\sup_{\omega \in [-\pi,\pi)} \|\hat{S}^{c}(\omega) - S(\omega)\| \;\xrightarrow{p}\; 0.
\]
Since the true spectral density $S(\omega)$ is uniformly bounded from below and above by positive constants, it is uniformly positive definite. Matrix inversion is continuous on this set, yielding
\[
\sup_{\omega \in [-\pi,\pi)} \|\hat{S}^{c}(\omega)^{-1} - S(\omega)^{-1}\| \;\xrightarrow{p}\; 0.
\]

Let $\hat{\Theta}(\omega)$ denote the GLASSO estimator applied frequency-wise to $\hat{S}^{c}(\omega)$ with regularization parameter $\lambda_T \to 0$. Since the penalized likelihood objective is strictly convex and smooth on the space of uniformly positive definite matrices, the GLASSO mapping is continuous and reduces to the unregularized inverse as $\lambda_T \to 0$. Consequently,
\[
\sup_{\omega \in [-\pi,\pi)} \|\hat{\Theta}(\omega)^{-1} - S(\omega)\| \;\xrightarrow{p}\; 0.
\]

Because $\mathbf{X}$ is purely non-deterministic, $S(\omega)$ admits a unique canonical (minimum-phase) spectral factorization $S(\omega) = H(\omega) \Sigma H(\omega)^H$ with $H(\omega)$ invertible and continuous in $\omega$. Canonical spectral factorization is continuous with respect to uniform perturbations of the spectral density on compact frequency sets, so the factor $\hat{H}(\omega)$ obtained from $\hat{\Theta}(\omega)^{-1}$ satisfies
\[
\sup_{\omega \in [-\pi,\pi)} \|\hat{H}(\omega)^{-1} - H(\omega)^{-1}\| \;\xrightarrow{p}\; 0.
\]

The non-parametric partial directed coherence $\pi_{k\leftarrow \ell}(\omega)$ is a continuous functional of the inverse spectral factor $\bar{A}(\omega) = H(\omega)^{-1}$, involving only matrix entries, absolute values, and normalization by a denominator bounded away from zero. By the continuous mapping theorem,
\[
\sup_{\omega \in [-\pi,\pi)} \bigl|\hat{\pi}_{k\leftarrow \ell}(\omega) - \pi_{k\leftarrow \ell}(\omega)\bigr| \;\xrightarrow{p}\; 0
\quad \text{for all } k, \ell \in [K].
\]
\end{proof}

\section{Sufficient conditions for completion consistency (Proposition~\ref{prop:nnm-sufficient})}

\begin{proposition}[Sufficient conditions for completion consistency via nuclear-norm minimization]
\label{prop:nnm-sufficient}
Suppose $S(\omega)$ admits the decomposition~\eqref{eq:lowrank-plus-noise} with effective rank $r$. Suppose further that:
\begin{enumerate}[leftmargin=2em]
  \item[(i)] (Spikiness) $\|S(\omega)\|_\infty \;\lesssim\; K^{-1/2}\, \|S(\omega)\|_F$ uniformly in $\omega$;
  \item[(ii)] (Coverage) the set $\mathcal{O}$ of observed pairs satisfies
  \[
  |\mathcal{O}| \;\gtrsim\; r\, K\, \log K \cdot \log T,
  \]
  where $T = \min_u T_u$;
  \item[(iii)] (Local consistency) Assumption~\ref{ass:local} holds.
\end{enumerate}
Then, for a penalty parameter $\mu>0$, the frequency-wise nuclear-norm estimator
\[
\hat{S}^{c}(\omega) \;=\; \arg\min_{M \succeq 0}\; \|M\|_* \;+\; \frac{\mu}{2}\,\bigl\|\mathcal{P}_{\mathcal{O}}\bigl(M - \hat{S}(\omega)\bigr)\bigr\|_F^2
\]
satisfies~\eqref{eq:completion-consistency}, so that the conclusion of Theorem~\ref{thm:incomplete} holds.
\end{proposition}

The result follows from the matrix completion bounds of Negahban and Wainwright~\cite{negahban2012restricted}, applied frequency-wise. Their main theorem establishes that, under the spikiness condition $\|S(\omega)\|_\infty \lesssim K^{-1/2} \|S(\omega)\|_F$ and the coverage condition $|\mathcal{O}| \gtrsim r K \log K$, the nuclear-norm penalized least-squares estimator recovers a near-low-rank matrix from noisy partial observations, with mean-squared Frobenius error decomposing into a stochastic term that vanishes as the noise variance shrinks and an approximation term that vanishes as the effective rank captures the matrix energy. In our setting, the noise on each observed entry is the local PSD estimation error $\hat{S}_{ij}(\omega) - S_{ij}(\omega)$, whose variance vanishes as $T_u \to \infty$ under Assumption~\ref{ass:local}. The approximation term vanishes under the near-low-rank decomposition~\eqref{eq:lowrank-plus-noise}. To obtain uniform consistency across the frequency grid of size $\asymp T$, a union bound contributes the additional $\log T$ factor that appears in coverage condition (ii). Translating $|\mathcal{O}|$ to a per-pair observation probability $p_{\mathrm{obs}}$ gives~\eqref{eq:coverage-prob}. The conclusion of Theorem~\ref{thm:incomplete} then follows directly from the verified hypothesis~\eqref{eq:completion-consistency}.

\section{Izhikevich spiking neural network: model parameters and simulation details}
\label{S1_Text}

\subsection{Network architecture}

The biophysically realistic validation network consists of $N = 3$ neurons: two
excitatory regular-spiking units $N_1(\text{E})$ and $N_2(\text{E})$, and one
inhibitory fast-spiking unit $N_3(\text{I})$.
All-to-all directed synaptic connectivity is encoded in the weight matrix
(columns: pre-synaptic; rows: post-synaptic)
\begin{equation}
W =
\begin{pmatrix}
 0   & 10  & -18 \\
 14  &  0  & -16 \\
  9  & 12  &   0
\end{pmatrix},
\end{equation}
where positive entries denote excitatory synapses and negative entries denote
inhibitory synapses.
Diagonal entries are zero (no self-connections).

\subsection{Single-neuron dynamics}

Each neuron $k \in \{1, 2, 3\}$ follows the Izhikevich two-variable model~\cite{izhikevich2003simple}:
\begin{align}
\frac{dv_k}{dt} &= 0.04\,v_k^2 + 5\,v_k + 140 - u_k + I_k(t), \\
\frac{du_k}{dt} &= a_k\bigl(b_k\,v_k - u_k\bigr),
\end{align}
with the reset rule: if $v_k \geq 30\,\text{mV}$, then $v_k \leftarrow c_k$
and $u_k \leftarrow u_k + d_k$.
Here $v_k$ is the membrane potential (mV) and $u_k$ is a recovery variable.
The neuron-specific parameters are given in Table~\ref{tab:izh_params}.

\begin{table}[h!]
\centering
\caption{Izhikevich model parameters for each neuron.}
\label{tab:izh_params}
\begin{tabular}{lcccc}
\hline
\textbf{Neuron} & $a$ & $b$ & $c$ (mV) & $d$ \\
\hline
$N_1(\text{E})$ & 0.02 & 0.200 & $-64.9$ & 1.0 \\
$N_2(\text{E})$ & 0.02 & 0.200 & $-63.7$ & 1.0 \\
$N_3(\text{I})$ & 0.06 & 0.225 & $-65.0$ & 1.0 \\
\hline
\end{tabular}
\end{table}

The parameters for $N_1(\text{E})$ and $N_2(\text{E})$ correspond to the
regular-spiking excitatory cell type; the elevated time-scale parameter $a$
and recovery sensitivity $b$ for $N_3(\text{I})$ correspond to a
fast-spiking inhibitory interneuron~\cite{izhikevich2003simple}.

\subsection{External input}

Each neuron receives independent stochastic input
\begin{equation}
I_k(t) = \mu + \sigma\,\xi_k(t),
\end{equation}
where $\mu = 4.0$ (mean drive), $\sigma = 18.0$ (noise amplitude), and
$\xi_k(t)$ is independent standard Gaussian white noise sampled at each
millisecond time step.
The mean drive $\mu$ is chosen to maintain spontaneous firing in all three
neurons; the noise amplitude $\sigma$ ensures irregular, non-periodic spiking
consistent with in-vivo-like activity.

\subsection{Simulation}

The network was simulated for $T_{\mathrm{sim}} = 200{,}000\,\text{ms}$
(200 s) using the standard Euler integration scheme with a 1 ms time step.
Spike times were recorded whenever $v_k(t) \geq 30\,\text{mV}$.
Synaptic input from neuron $j$ to neuron $k$ was delivered as an
instantaneous additive current $W_{kj}$ at the time of each pre-synaptic spike,
with no synaptic delay.

\section{IBL Neuropixels sessions and brain areas retained for SPIDER analysis}
\label{S2_Text}

\subsection{Overview}

A total of 43 recording sessions from the International Brain Laboratory (IBL)
brain-wide map dataset were retained after quality control and session-selection
procedures described in the Methods.
Sessions span 12 contributing laboratories and 39 individual adult mice (four mice each contributed two sessions).
Each session is identified by the format
\texttt{lab\_\_subject\_\_date\_\_session-number}.
Together the 43 sessions cover 50 distinct brain areas after Allen CCF
hierarchical remapping.

\subsection{Retained sessions}

\begin{table}[h]
\centering
\small
\caption{IBL sessions retained for SPIDER analysis (43 sessions,
12 laboratories).}
\label{tab:sessions}
\begin{tabular}{lllll}
\hline
\textbf{\#} & \textbf{Laboratory} & \textbf{Subject} & \textbf{Date} &
\textbf{Areas recorded} \\
\hline
1  & angelakilab       & NYU-11       & 2020-02-21 & 9  \\
2  & angelakilab       & NYU-37       & 2021-01-25 & 6  \\
3  & angelakilab       & NYU-45       & 2021-07-19 & 6  \\
4  & churchlandlab     & CSHL049      & 2020-01-09 & 4  \\
5  & churchlandlab     & CSHL051      & 2020-02-05 & 5  \\
6  & churchlandlab     & CSHL055      & 2020-02-17 & 6  \\
7  & churchlandlab     & CSHL058      & 2020-07-07 & 4  \\
8  & churchlandlab     & CSHL058      & 2020-07-08 & 5  \\
9  & churchlandlab\_ucla & MFD\_05    & 2023-08-16 & 4  \\
10 & churchlandlab\_ucla & MFD\_07    & 2023-08-31 & 4  \\
11 & churchlandlab\_ucla & UCLA033    & 2022-02-15 & 9  \\
12 & churchlandlab\_ucla & UCLA036    & 2022-03-09 & 7  \\
13 & churchlandlab\_ucla & UCLA037    & 2022-02-02 & 7  \\
14 & cortexlab         & KS044        & 2020-12-09 & 9  \\
15 & cortexlab         & KS046        & 2020-12-04 & 7  \\
16 & cortexlab         & KS051        & 2021-05-11 & 5  \\
17 & cortexlab         & KS096        & 2022-06-17 & 4  \\
18 & danlab            & DY\_008      & 2020-03-03 & 5  \\
19 & danlab            & DY\_009      & 2020-02-27 & 6  \\
20 & danlab            & DY\_014      & 2020-07-19 & 6  \\
21 & danlab            & DY\_016      & 2020-09-11 & 5  \\
22 & danlab            & DY\_016      & 2020-09-12 & 6  \\
23 & danlab            & DY\_020      & 2020-09-29 & 9  \\
24 & hausserlab        & PL017        & 2021-11-14 & 4  \\
25 & hoferlab          & SWC\_061     & 2020-11-26 & 5  \\
26 & mainenlab         & ZFM-01935    & 2021-02-03 & 7  \\
27 & mainenlab         & ZFM-01936    & 2021-01-23 & 10 \\
28 & mainenlab         & ZFM-01937    & 2021-04-08 & 4  \\
29 & mainenlab         & ZFM-02370    & 2021-04-28 & 4  \\
30 & mainenlab         & ZFM-02373    & 2021-06-23 & 5  \\
31 & mainenlab         & ZM\_2240     & 2020-01-22 & 7  \\
32 & mainenlab         & ZM\_2240     & 2020-01-23 & 4  \\
33 & mrsicflogellab    & SWC\_038     & 2020-07-30 & 6  \\
34 & mrsicflogellab    & SWC\_054     & 2020-10-05 & 5  \\
35 & steinmetzlab      & NR\_0027     & 2022-08-19 & 4  \\
36 & steinmetzlab      & NR\_0028     & 2023-03-07 & 4  \\
37 & steinmetzlab      & NR\_0028     & 2023-03-15 & 8  \\
38 & wittenlab         & ibl\_witten\_27 & 2021-01-21 & 7 \\
39 & zadorlab          & CSH\_ZAD\_011 & 2020-03-23 & 5 \\
40 & zadorlab          & CSH\_ZAD\_022 & 2020-05-24 & 6 \\
41 & zadorlab          & CSH\_ZAD\_025 & 2020-08-12 & 9 \\
42 & zadorlab          & CSH\_ZAD\_026 & 2020-08-14 & 7 \\
43 & zadorlab          & CSH\_ZAD\_029 & 2020-09-19 & 6 \\
\hline
\end{tabular}
\end{table}

\subsection{Brain areas per session}

Table~\ref{tab:areas} lists the brain areas recorded in each session after
Allen CCF hierarchical remapping and quality control filtering.
Area acronyms follow the Allen Mouse Brain Common Coordinate Framework (CCF)
v3 nomenclature.

\begin{table}[h]
\centering
\small
\caption{Brain areas retained per session after quality control.}
\label{tab:areas}
\begin{tabular}{lp{11cm}}
\hline
\textbf{Session \#} & \textbf{Areas} \\
\hline
1  & BMA, CA1, CA3, CEA, COA, GPe, LGd, TH, VPM \\
2  & CP, DG, Eth, LP, SGN, SSp \\
3  & APN, CA1, DG, MB, PO, VISa \\
4  & DG, Eth, LP, VISa \\
5  & APN, CA1, DG, VISa, VISam \\
6  & CA1, CP, SSp, TH, VPL, VPM \\
7  & CA1, PO, VISa, VISam \\
8  & ILA, MOs, OLF, PL, STR \\
9  & ProS, RN, SCm, VISp \\
10 & CA1, DG, PO, VISa \\
11 & CA1, MB, MRN, ProS, RN, SCm, SUB, VISam, VISp \\
12 & APN, CA1, DG, LP, POL, VISa, VISam \\
13 & APN, CA1, DG, Eth, MRN, TH, VISa \\
14 & CA1, COA, MEA, MG, PoT, SGN, VISa, VPL, VPM \\
15 & ACA, AON, ILA, MOs, OLF, ORB, PL \\
16 & CA1, DG, LP, PoT, VISa \\
17 & CA1, DG, LP, VISa \\
18 & CA1, CA3, LP, SSp, VPM \\
19 & CA1, DG, LP, PoT, SGN, VISp \\
20 & ACB, CA3, CP, LP, SI, SSp \\
21 & CA1, CA3, LGd, TH, VISa \\
22 & CA1, DG, Eth, LP, PO, VISam \\
23 & CA1, CA3, LGd, MRN, RN, RSPv, SCm, VISa, VPL \\
24 & MB, MRN, RSPv, SCm \\
25 & AON, CP, MOs, OLF, ORB \\
26 & CA1, CA3, Eth, LD, LGd, SSp, VPM \\
27 & ACA, ACB, ILA, LS, MD, MOs, PL, PVT, SI, STR \\
28 & ACA, LS, MOp, MOs \\
29 & ACA, LS, MOs, STR \\
30 & CA1, CA3, DG, VISa, VPM \\
31 & ACA, ILA, LS, MOs, OLF, PL, STR \\
32 & CA1, DG, RSPagl, SSp \\
33 & AON, ILA, MOs, OLF, ORB, PL \\
34 & CA1, DG, Eth, PO, VISa \\
35 & APN, CA1, DG, MRN \\
36 & MRN, ProS, SCm, SUB \\
37 & CA1, CA3, DG, PAG, RSPagl, RSPv, SCm, VISp \\
38 & CA1, CA3, DG, MG, POL, TH, VISam \\
39 & CA3, DG, MRN, ProS, SCm \\
40 & APN, CA1, CA3, LGd, VISa, ZI \\
41 & CA1, CA3, DG, LP, PO, RSPagl, VPL, VPM, ZI \\
42 & CP, MRN, OLF, SSp, STR, SUB, VISp \\
43 & BMA, CEA, CP, OLF, SSp, STR \\
\hline
\end{tabular}
\end{table}

\subsection{Complete list of retained brain areas (50 areas)}

The following 50 brain areas were retained across all sessions and used for
the final SPIDER estimation.
Areas are listed alphabetically with their full Allen CCF names.

\begin{table}[h]
\centering
\small
\caption{Brain areas retained in the final analysis.}
\label{tab:allAreas}
\begin{tabular}{lll}
\hline
\textbf{Acronym} & \textbf{Full name} & \textbf{Major division} \\
\hline
ACA    & Anterior cingulate area                    & Isocortex \\
ACB    & Nucleus accumbens                          & Striatum \\
AON    & Anterior olfactory nucleus                 & Olfactory \\
APN    & Anterior pretectal nucleus                 & Midbrain \\
BMA    & Basomedial amygdala                        & Amygdala \\
CA1    & Field CA1                                  & Hippocampus \\
CA3    & Field CA3                                  & Hippocampus \\
CEA    & Central amygdalar nucleus                  & Amygdala \\
COA    & Cortical amygdalar area                    & Amygdala \\
CP     & Caudoputamen                               & Striatum \\
DG     & Dentate gyrus                              & Hippocampus \\
Eth    & Ethmoid nucleus                            & Thalamus \\
GPe    & Globus pallidus, external                  & Pallidum \\
ILA    & Infralimbic area                           & Isocortex \\
LD     & Lateral dorsal nucleus of thalamus         & Thalamus \\
LGd    & Dorsal part of lateral geniculate complex  & Thalamus \\
LP     & Lateral posterior nucleus of thalamus      & Thalamus \\
LS     & Lateral septal nucleus                     & Cerebral nuclei \\
MB     & Midbrain                                   & Midbrain \\
MD     & Mediodorsal nucleus of thalamus            & Thalamus \\
MEA    & Medial amygdalar nucleus                   & Amygdala \\
MG     & Medial geniculate complex                  & Thalamus \\
MOp    & Primary motor area                         & Isocortex \\
MOs    & Secondary motor area                       & Isocortex \\
MRN    & Midbrain reticular nucleus                 & Midbrain \\
OLF    & Olfactory areas                            & Olfactory \\
ORB    & Orbital area                               & Isocortex \\
PAG    & Periaqueductal gray                        & Midbrain \\
PL     & Prelimbic area                             & Isocortex \\
PO     & Posterior complex of thalamus              & Thalamus \\
POL    & Posterior limiting nucleus of thalamus     & Thalamus \\
PoT    & Posterior triangular nucleus of thalamus   & Thalamus \\
ProS   & Prosubiculum                               & Hippocampus \\
PVT    & Paraventricular nucleus of thalamus        & Thalamus \\
RN     & Red nucleus                                & Midbrain \\
RSPagl & Retrosplenial area, lateral agranular      & Isocortex \\
RSPv   & Retrosplenial area, ventral                & Isocortex \\
SCm    & Superior colliculus, motor related         & Midbrain \\
SGN    & Suprageniculate nucleus                    & Thalamus \\
SI     & Substantia innominata                      & Pallidum \\
SSp    & Primary somatosensory area                 & Isocortex \\
STR    & Striatum                                   & Striatum \\
SUB    & Subiculum                                  & Hippocampus \\
TH     & Thalamus                                   & Thalamus \\
VISa   & Visual area anterior                       & Isocortex \\
VISam  & Anteromedial visual area                   & Isocortex \\
VISp   & Primary visual area                        & Isocortex \\
VPL    & Ventral posterolateral nucleus             & Thalamus \\
VPM    & Ventral posteromedial nucleus              & Thalamus \\
ZI     & Zona incerta                               & Hypothalamus \\
\hline
\end{tabular}
\end{table}

\end{document}